\documentclass[11pt]{article}

\usepackage{amsfonts}
\usepackage{amssymb,amsmath,amsthm}
\usepackage{epsfig}
\usepackage{color}
\usepackage{latexsym}
\usepackage{enumerate}
\usepackage{algorithm}
\usepackage{algorithmic}

\usepackage{pgfplots}
\pgfplotsset{compat=1.4}
\usetikzlibrary{patterns}

\usepackage{tweaklist}

\renewcommand{\enumhook}{\setlength{\topsep}{2pt}%
  \setlength{\itemsep}{-2pt}}
\renewcommand{\itemhook}{\setlength{\topsep}{2pt}%
  \setlength{\itemsep}{-2pt}}
\usepackage[dvips, paper=letterpaper, top=1in, bottom=.75in, left=1in, right=1in, nohead, includefoot, footskip=.25in]{geometry}

\newtheorem{theorem}{Theorem}
\newtheorem{corollary}{Corollary}
\newtheorem{lemma}{Lemma}

\newtheorem{claim}{Claim}

\newtheorem{definition}{Definition}

\newtheorem{example}{Example}

\def \citeN {\cite}

\newcommand{\notshow}[1]{{}}

\definecolor{MyGray}{rgb}{0.8,0.8,0.8}

\makeatletter
\newtheorem*{rep@theorem}{\rep@title}
\newcommand{\newreptheorem}[2]{%
\newenvironment{rep#1}[1]{%
 \def\rep@title{#2 \ref{##1}}%
 \begin{rep@theorem}}%
 {\end{rep@theorem}}}
\makeatother

\newreptheorem{theorem}{Theorem}

\newcommand{\bR}{\mathbb{R}}

\newcommand{\cM}{{\cal M}}

\newcommand{\cP}{{\cal P}}

\newcommand{\cT}{{\cal T}}
\newcommand{\cU}{{\cal U}}

\begin{document}

\title{Mechanism Design via Optimal Transport}

\author {
Constantinos Daskalakis\thanks{Supported by a Sloan Foundation Fellowship, a Microsoft Research Faculty Fellowship, and NSF Award CCF-0953960 (CAREER) and CCF-1101491.}\\
EECS, MIT \\
\tt{costis@mit.edu}
\and
Alan Deckelbaum\thanks{Supported by Fannie and John Hertz Foundation Daniel Stroock Fellowship and NSF Award CCF-1101491.}\\
Math, MIT\\
\tt{deckel@mit.edu}
\and
Christos Tzamos\thanks{Supported by NSF Award CCF- 1101491.}\\
EECS, MIT\\
\tt{tzamos@mit.edu}
}
\addtocounter{page}{-1}
\maketitle

\begin{abstract}
Optimal mechanisms have  been  provided in quite general multi-item settings~\cite{CaiDW12b}, as long as each bidder's type distribution is given explicitly by listing every type in the support along with its associated probability. In the implicit setting, e.g. when the bidders have additive valuations with independent  and/or continuous values for the items, these results do not apply, and it was recently shown that exact revenue optimization is intractable, even when there is only one bidder~\cite{DaskalakisDT13}. Even for item distributions with special structure, optimal mechanisms have been surprisingly rare~\cite{ManelliV06} and the problem is challenging even in the two-item case~\cite{HartN12}. In this paper, we provide a  framework for designing optimal mechanisms using optimal transport theory and duality theory. We instantiate our framework to obtain conditions under which only pricing  the grand bundle is optimal in multi-item settings (complementing the work of~\cite{ManelliV06}), as well as to characterize optimal two-item mechanisms. We use our results to derive closed-form descriptions of the optimal mechanism in several two-item settings, exhibiting also a setting where a continuum of lotteries is necessary for revenue optimization but a closed-form representation of the mechanism can still be found efficiently using our framework.
\end{abstract}

\thispagestyle{empty}

\newpage

\section{Introduction}\label{sec:intro}

{\em Optimal mechanism design} is the problem of designing a revenue-optimal auction for selling $n$ items to $m$ bidders whose valuations are drawn from known prior distributions. The special case of selling a single item is well-understood, going back to the work of ~\citeN{Myerson81} and \citeN{CM88}. The general ($n>1$) case has been much more challenging, and until very recently there has been no general solution. In a series of recent papers, Cai et al. provided efficiently computable revenue-optimal~\cite{CaiDW12,CaiDW12b} or approximately optimal mechanisms~\cite{CaiDW13} in very general settings, including when there are combinatorial constraints over which allocations of items to bidders are feasible, e.g. when these are matroid, matching, or more general constraints. However, these results, as well as the more specialized ones of~\cite{AlaeiFHHM12} for service-constrained environments, apply to the {\em explicit setting}, i.e. when the distributions over bidders' valuations are given explicitly, by listing every valuation in their support together with the probability it appears. 

Clearly, the explicit is not the right model when the type distributions are continuous and/or have extra structure that allows for more succinct representation. The obvious example, and the setting that we study in this paper, is when the bidders have {\em additive valuations} with {\em independent values} for different items. Here, each bidder's type distribution can be described by providing one marginal distribution per item, saving an exponential amount of information compared to the explicit description. The issue is that such {\em implicit settings} turn out to be even more challenging computationally. Indeed, essentially the only known positive results for additive bidders in the implicit setting are for when the values are drawn from Monotone Hazard Rate distributions where~\cite{BhattacharyaGGM10} obtain constant factor approximations to the optimal revenue, and~\cite{DaskalakisW12,CaiH13} obtain polynomial-time approximation schemes. For general distributions but a single buyer, \cite{HartN12} show that selling the items through separate auctions guarantees a $O({1 \over \log^2 n})$-fraction of the optimal revenue, which can be improved to ${1 \over 2}$ for $2$ items, even if the number of buyers is arbitrary. They also show that in the single-buyer setting with identically distributed items, offering the grand bundle at some optimal price guarantees a $O(\frac {1}{\log n})$-fraction of the optimal revenue.  At the same time, exact polynomial-time solutions have been recently precluded by~\cite{DaskalakisDT13}, where it is shown that computing optimal mechanisms is $\#$P hard, even when there is a single additive bidder whose values for the items are independent of support~$2$.

\smallskip The scarcity of algorithmic results as well as the recent computational lower bound~\cite{DaskalakisDT13} are consistent with our lack of structural understanding of the optimal mechanism in this setting. It had been long known that selling the items separately is sub-optimal. Here is an example from~\cite{HartN12}: Suppose that there are two items and an additive bidder whose values are independent and uniformly distributed in $\{1,2\}$. It is easy to see that selling each item separately results in expected revenue at most $2$, while if the auctioneer only offers the bundle of both items for $3$, the expected revenue is $2.25$. So bundling the items increases  revenue. It is also known that, unlike  the single-item case, the optimal mechanism need not be deterministic~\cite{Thanassoulis04,ManelliV06,ManelliV07,HartR11}. Here is an example from~\cite{DaskalakisDT13}: Suppose  there are two items and an additive bidder whose values are independent and uniformly distributed in $\{1,2\}$ and $\{1,3\}$ respectively. In this scenario, the optimal mechanism offers the bundle of both the items at price $4$; it also offers at price $2.5$ a lottery that, with probability $1/2$, gives both  items and, with probability $1/2$, offers just the~first~item. 

Besides these two insights (that bundling and randomization help) surprisingly little is known about the structure of the optimal mechanism, even in the single-bidder case that we consider in this paper. \citeN{Mcafee88} proposed conditions under which the optimal mechanism is deterministic, however these were found insufficient by \citeN{Thanassoulis04} and \citeN{ManelliV06}. The advantage of deterministic mechanisms is that they have a finite description: the price they charge for every possible bundle of the items. Hence looking for the optimal one is feasible, computational considerations aside. On the other hand, randomization adds an extra layer of difficulty: it is possible---we exhibit such an example in Section~\ref{sec:beta}---that the optimal mechanism offers a continuum of lotteries. Hence it is a priori not clear whether one could hope for  a concise (even a finite) description of the optimal mechanism, and it is even less clear whether one can optimize over the corresponding space of (infinite-dimensional) mechanisms. 

\smallskip
In this paper, we develop a general optimization framework for obtaining closed-form descriptions of optimal mechanisms in multi-item settings with one additive buyer with independent values for the items, where each value is distributed according to a continuous distribution specified by a closed-form description of its probability density function. Our framework is outlined below:

\begin{enumerate}

\item Optimal mechanism design in our setting is known to be reducible to optimizing a specific integral objective of the utility function $u: \mathbb{R}_{\geq 0}^n \rightarrow \mathbb{R}$  of the buyer, with the constraint that $u$ is increasing, convex, and continuous with gradient $\nabla u \in [0,1]^n$ almost everywhere; see e.g.~\cite{ManelliV06}. We describe this formulation in Section~\ref{sec:prelims}. \label{step1: primal}

\item Our first step is to relax this optimization problem (Section~\ref{sec:relaxedprogram}) by relaxing the convexity constraint. This constraint is intimately related to the truthfulness of the resulting mechanism; and, when violated, truthfulness and as a consequence revenue optimality are at stake. Regardless, we relax this constraint for developing our framework, and restore it only for our end results.\label{step2: relaxed primal}

\item We provide a dual to the relaxed problem, which amounts to the following optimal transport problem:

\begin{itemize}
\item {\sc Input:} Two probability measures $\mu$ and $\nu$ on $\mathbb{R}_{\geq 0}^n$, implicitly defined by the buyer's value distributions;

\item {\sc Output:} The optimal transport from $\mu$ to $\nu$, where the cost of transferring a unit of probability mass from point $x$ to point $y$ is $\sum_{i} \max\{x_i-y_i,0\}$.

\end{itemize}
This dual problem is a continuous analog of minimum-weight bipartite matching. We prove a complementary slackness condition (Theorem~\ref{thm::main}) for certifying optimality of primal and dual solutions.
\label{step3: dual}

\item  The end products of our dual relaxation, complementary slackness, and restoration of convexity are two structural theorems for optimal mechanism design:
\begin{itemize}

\item Theorem~\ref{bundlingthm} provides a condition under which the optimal mechanism is a take-it-or-leave-it offer of the grand bundle at some critical price $p^*$; $p^*$ is not arbitrary but the boundary of a convex region of a particular measure, which we show how to define in Section~\ref{bundlingsection}. Instantiating Theorem~\ref{bundlingthm},
\begin{itemize}
\item  {\em (exponential distributions)} we analytically derive the optimal mechanism for two items distributed according to exponential distributions with arbitrary parameters $\lambda_1$, $\lambda_2$; we show that the optimal mechanism offers two options: (i) the grand bundle at some price and (ii) at a different price, the item with the thinner tail with probability $1$ and the other item with probability $\min\{\lambda_i\} \over \max\{\lambda_i\}$; see Section~\ref{exponentialsolution}.

\item {\em (power-law distributions)} we exhibit a setting with two items distributed according to non-identical power law distributions where the optimal mechanism only offers  the grand bundle; see Section~\ref{powerlaw}.

\end{itemize}

\item Theorem~\ref{generalthm} provides our general characterization of the optimal mechanism for two items. The characterization applies to settings that are {\em canonical} according to Definition~ \ref{def:canonical2}. Under this condition, the optimal mechanism is succinctly described in terms of a  decreasing, concave and continuous function in $\mathbb{R}^2$: all types under this function are allocated nothing and pay nothing; all other types are matched to a point of the function, in some canonical way specified by Theorem~\ref{generalthm}, and their allocation probabilities correspond to derivatives of the function at the corresponding point. Using our general theorem,
\begin{itemize}
\item {\em (beta distributions)} we exhibit a setting with two items distributed according to non-identical beta distributions where the optimal mechanism offers a continuum of lotteries; see Example~\ref{betaexample} in Section~\ref{sec:beta}. To the best of our knowledge, this is the first known explicit setting with two independent values where the optimal mechanism comprises a continuum of lotteries.
\end{itemize}
\end{itemize}

\item In the proofs of our structural theorems, we employ Strassen's theorem on stochastic domination of measures~\cite{Lindvall99}. As a consequence of our proof technique, we introduce a condition on stochastic domination in both our results. As it may be cumbersome to check stochastic domination, we develop an alternate condition (see Theorem~\ref{regionthm}) that implies stochastic domination. Our new condition is of independent interest to measure theory, and will be useful to the user of our results. Indeed, we rely on it for all our $2$-item applications described above. 
\end{enumerate}

\section{The Revenue-Maximization Program}\label{sec:prelims}

We aim to find the revenue-optimal mechanism $\cM$ for selling $n$ goods to a single additive bidder whose values $z = (z_1,...,z_n)$ for the goods are drawn independently from probability distributions with given densities $f_i(z_i): \mathbb{R}_{\geq 0} \rightarrow \mathbb{R}_{\geq 0}$, for all $i$. $\cM$ takes as input the vector $z$ of bidder's values and outputs the probability with which he will receive each good along with the price that he needs to pay for this  allocation. That is, $\cM$ consists of two functions $\cP: \bR_{\geq 0}^n \rightarrow [0,1]^n$ and $\cT: \bR_{\geq 0}^n \rightarrow \bR$ that give, respectively, the vector of allocation probabilities and the price that the bidder pays, as a function of the bidder's values.

The bidder receives utility $\cU(z, p,t) = z \cdot p - t$ when his values for the items are $z$ and he is offered the items with probabilities $p$ at price $t$.

We restrict our attention to mechanisms that are \emph{incentive compatible}, meaning that the bidder must have adequate incentives to reveal his values for the items truthfully, and \emph{individually rational}, meaning that the bidder has an incentive to participate in the mechanism.

\begin{definition}
The mechanism $\cM$ is \emph{incentive compatible (IC)} if and only if $\cU(z, \cP(z),  \cT(z)) \ge \cU(z, \cP(z'),  \cT(z'))$ for all $z,z' \in \mathbb{R}^n_{\geq 0}$.
\end{definition}

\begin{definition}
The mechanism $\cM$ is \emph{individually rational (IR)} if and only if $\cU(z, \cP(z),  \cT(z)) \ge 0$ for all $z \in \mathbb{R}^n_{\geq 0}$.
\end{definition}

When we enforce the IC constraints, we can assume that the buyer truthfully reports his type to the mechanism. Under this assumption, let $u: \mathbb{R}_{\geq 0}^n \rightarrow \mathbb{R}$ be the function that maps the buyer's valuation to the utility he receives by the mechanism $\cM$. We have that $u(z) = \cU(z, \cP(z),  \cT(z))$.

As noted in  \citeN{ManelliV06}, the mechanism can be uniquely determined by providing its corresponding utility function $u$. We readily use the characterization that the mechanism is IC if and only if $u$ is non-decreasing, convex, continuous, and $\nabla u \in [0,1]^n$ almost everywhere. We additionally require that $u(z) \ge 0$ to satisfy the IR constraint. Given $u$ we can compute the functions $\cP$ and $\cT$ by using the fact that $\cP(z) = \nabla u(z)$ and $\cT(z) =  \nabla u(z) \cdot z - u(z)$. 

Therefore, to find the revenue-optimal mechanism, we need to determine, given the probability density function $f(z)= \prod_i f_i(z_i)$, a nonnegative, nondecreasing, convex and continuous function $u$ with $\nabla u \in [0,1]^n$ almost everywhere such that the integral shown below (which equals the expected revenue) is maximized:
\begin{align}\int_{\bR_{\geq 0}^n} \cT(z) f(z) dz = \int_{\bR_{\geq 0}^n} \left[ \nabla u(z) \cdot z - u(z) \right] f(z) dz. \label{eq: revenue objective}
\end{align}

Throughout the paper, we make the following assumptions about the distribution of the buyer's value for each item:
\begin{itemize}
	\item The points where $f_i(z_i)$ is strictly positive lie within
	%\footnote{While the results in this paper can be applied to many scenarios where $f_i(x_i)$ is 0 in $D_i$ without much additional difficulty, we assume $f_i(x_i) > 0$  for  simplicity.} 
	 a semi-open (not necessarily bounded) interval $D_i \triangleq [d_i^-, d_i^+)$, where $d_i^-$ is nonnegative and $d_i^+$ is possibly infinite.
	\item $f_i(z_i)$ is continuously differentiable on $D_i$.
	\item $d_i^- f_i(d_i^-)  = 0$.%$\lim_{x_i \rightarrow d_i^-}x_i f_i(x_i) = 0$.
	\item $\lim_{z_i \rightarrow d_i^+} z_i^2 f_i(z_i) = 0$ .
	%\item $x_i f_i'(x_i)/f_i(x_i)$ is a decreasing function of $x_i$ on $D_i \setminus \{z: f_i(z) = 0\}$.
\end{itemize}% (a continuously differentiable real-valued function for all $x_i \geq 0$) and suppose that $\lim_{x_i \rightarrow \infty} x_i^2 f_i(x_i) = 0$ for all $i$.
We denote by $D \subseteq \mathbb{R}^n_{\geq 0}$ the set $\times_i D_i$. We now explicitly write down the expected revenue:
$$\int_{d_n^-}^{d_n^+} \cdots \int_{d_1^-}^{d_1^+} \left( z_1 \frac{\partial u}{\partial z_i} + \cdots + z_n \frac{\partial u}{\partial z_n} - u(z) \right)f_1(z_1)\cdots f_n(z_n) dz_1\cdots dz_n.$$
Integrating by parts, we see that
$$\int_{d_i^-}^{d_i^+} z_i \frac{\partial u}{\partial z_i}f_i(z_i)dz_i = \lim_{M \rightarrow d_i^+} u(z_i,z_{-i})z_if_i(z_i)|_{z_i = d_i^-}^{z_i = M} - \int_{d_i^-}^{d_i^+} u(z_i,z_{-i})(f_i(z_i) + z_if'_i(z_i))dz_i.$$
Since $u(z_i, z_{-i}) \leq \sum_j z_j$ and $z_i^2f_i(z_i)\rightarrow 0$ as $z_i \rightarrow d_i^+$, we see that the first term is 0 for any fixed $z_{-i}$. Therefore, we have
$$\int_{d_i^-}^{d_i^+} z_i \frac{\partial u}{\partial z_i}f_i(z_i)dz_i = - \int_{d_i^-}^{d_i^+} u(z_i,z_{-i})(f_i(z_i) + z_if'_i(z_i))dz_i$$
and thus 
the integral giving the expected revenue for a chosen utility function $u$ is
\begin{align}
\int_D u(z)\left( -\nabla f(z) \cdot z - (n+1)f(z) \right) dz.\label{eq: revenue objective 2}
\end{align}

\section{Separating into Two Spaces}

We denote by $d_-$ the point $(d_1^-,\ldots,d_n^-) \in D$. For reasons which will become clear later, we treat the point $d_-$ differently from the rest of $D$. The formulation of Section~\ref{sec:prelims} naturally separates $D \setminus \{d_-\}$ into two regions:
\begin{align*}
\mathcal{X} &= \left\{ z \in D:-\nabla f(z) \cdot z > (n+1)f(z)\right\} \setminus\{d_-\} \\
 \mathcal{Y} &= \left\{z \in D: -\nabla f(z) \cdot z \leq (n+1)f(z)\right\} \setminus\{d_-\}.
 \end{align*}
To maximize revenue, we aim for $u$ to be large on $\mathcal{X}$ yet small on $\mathcal{Y}$.
We define the density functions $\mu_d^{\mathcal{X}} : \mathcal{X} \rightarrow \mathbb{R}$ and $\nu_d^{\mathcal{Y}} : \mathcal{Y} \rightarrow \mathbb{R}$ by
$$\mu_d^{\mathcal{X}}(x) = -\nabla f(x) \cdot x - (n+1)f(x); \qquad \nu_d^\mathcal{Y}(y) =  (n+1)f(y) + \nabla f(y) \cdot y$$
%$$\mu_d^{\mathcal{X}}(x) =  \left( -\sum_i x_i \frac{f_i'(x_i)}{f_i(x_i)} - (n+1) \right)f(x); \qquad \nu_d^{\mathcal{Y}}(y) = \left(n+1+ \sum_i y_i \frac{f_i'(y_i)}{f_i(y_i)} \right)f(y)$$
and define the corresponding measures $\mu^{\mathcal{X}}$ and $\nu^{\mathcal{Y}}$ on $\mathcal{X}$ and $\mathcal{Y}$, respectively.
 
With this notation, our problem is to find the function $u: D \rightarrow \mathbb{R}$ maximizing
$$\int_\mathcal{X} u(x) \mu_d^{\mathcal{X}}(x) dx - \int_\mathcal{Y} u(y) \nu_d^{\mathcal{Y}}(y) dy$$
subject to the constraints mentioned in Section~\ref{sec:prelims}.\footnote{We note that any $u : D \rightarrow \mathbb{R}_{\geq 0}$ can be appropriately extended to $u : \mathbb{R}^n_{\geq 0} \rightarrow \mathbb{R}$ while preserving all constraints. We therefore restrict our attention to $u : D \rightarrow \mathbb{R}_{\geq 0}$.} We notice that $u(d_-) = 0$ in any optimal mechanism.\footnote{If $u(d_-) > 0$ then we could charge all players an additional $u(d_-)$, thereby subtracting $u(d_-)$ from the utility function everywhere and increasing revenue.} Furthermore, we will show momentarily that
$\mu^{\mathcal{X}}(\mathcal{X}) = \nu^{\mathcal{Y}}(\mathcal{Y}) - 1$.
%\footnote{This obviously must hold, since otherwise an arbitrarily large constant $u$ would make the expected revenue arbitrary large, which is clearly a nonsensical scenario.} %and $\mathcal{Y}$ do not have the same total mass, since $\int_\mathcal{X} \mu(x) dx < \int_\mathcal{Y} \nu(y) dy$.
For technical reasons, we desire for the masses of the spaces to be equal under their respective measures, and we therefore insert the point $d_-$ into $\mathcal{X}$ by defining the space $\mathcal{X}_0 \triangleq \mathcal{X} \cup \{d_-\}$, and extending $\mu^{\mathcal{X}}$ by setting $\mu^{\mathcal{X}_0}(\{d_-\}) = 1 $
and $\mu^{\mathcal{X}_0}(A) = \mu^\mathcal{X}(A)$ for all $A \subseteq \mathcal{X}$. 

\begin{claim}\label{equalsone}
$\mu^{\mathcal{X}_0}(\mathcal{X}_0) = \nu^{\mathcal{Y}}(\mathcal{Y}).$
\end{claim}
\begin{proof}
It suffices to show that 
$\int_\mathcal{Y} \nu^{\mathcal{Y}}(y) dy - \int_{\mathcal{X}} \mu^{\mathcal{X}}(x) dx = 1.$ Indeed, we have
\begin{align*}
\int_\mathcal{Y} \nu^{\mathcal{Y}}(y) dy - \int_{\mathcal{X}} \mu^{\mathcal{X}}(x) dx  &= -\int_{D} (-\nabla f(z) \cdot z - (n+1)f(z) )dz.
\end{align*}
We note that $\int_{D} (-\nabla f(z) \cdot z - (n+1)f(z) ) dz$ is the expected revenue under the constant utility function $u(z) = 1$. (To see this recall that~\eqref{eq: revenue objective 2} represents the expected revenue under utility function $u$.) The expected revenue in this case is $-1$ (see this by plugging $u(\cdot)=1$ into~\eqref{eq: revenue objective}, which is the alternative expression for expected revenue), and thus $\int_\mathcal{Y} \nu^{\mathcal{Y}}(y) dy - \int_{\mathcal{X}} \mu^{\mathcal{X}}(x) dx =  1$, as desired.
\end{proof}
%We have, by construction, .

In summary, our goal is to find the function $u : D \rightarrow \mathbb{R}$ which maximizes
$$\int_{\mathcal{X}_0} u(x) d\mu^{\mathcal{X}_0} - \int_{\mathcal{Y}} u(y) d\nu^{\mathcal{Y}}$$
subject to the constraints that $u$ is nondecreasing, convex, continuous and $\nabla u \in [0,1]^n$ almost everywhere, and $u(d_-) = 0$.

%We note that the constraint $u(d_-) = 0$ is since $\mu^{\mathcal{X}_0}(\mathcal{X}_0) = \nu^{\mathcal{Y}}(\mathcal{Y})$, and we can therefore add or subtract arbitrary constants from $u$ without affecting the objective value.

\begin{figure}[h!]
\vspace{-.4cm}
\begin{center}
\begin{tikzpicture}
\begin{axis}[width=2.2in, height=2.2in, ymin=1.1, ymax=3.3, xmin=0, xmax=1,  ytick pos=left, ytick={11.1}, yticklabels={$0$}, xtick={0},xticklabels={$d_-$}]
  \addplot+[color=gray, fill=gray!50, domain=0:2,samples=200,mark=none]
 %{1.9-1.5*x}
 {2.6-2.2*x*x+0.3*sin(1000*x)}
 \closedcycle;
\addplot[color=black, mark=*] coordinates{
(0,1.1)};
\node at (axis cs:0.25,1.8){$\mathcal{Y}$};
\node at (axis cs:0.7,2.4){$\mathcal{X}$};
\end{axis}
\end{tikzpicture}
\end{center}
\vspace{-.5cm}
\caption{$\mathcal{Y}$ is the region in which $-\nabla f(z) \cdot z - (n+1)f(z) \leq 0$ (excluding $d_-$), while $\mathcal{X}$ is the region in which the quantity is positive. We define $\mathcal{X}_0 = \mathcal{X} \cup \{d_-\}$, and set $\mu^{\mathcal{X}_0}(\{d_-\}) = \nu^{\mathcal{Y}}(\mathcal{Y}) - \mu^{\mathcal{X}}(\mathcal{X}).$ In all examples considered in this paper, $d_-$ is ``surrounded'' only by points in the $\mathcal{Y}$ region.}
\end{figure}
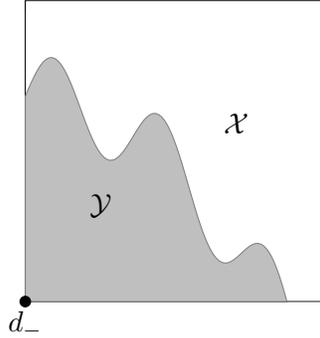
%Since $x_if_i(x_i)/f_i(x_i)$ is decreasing for each $i$, $\mathcal{Y}$ is a decreasing set. 
%\vspace{-.4cm}

\begin{example}[Exponential Distribution]
Suppose that the pdf of each item $i$ is given by $f_i(z_i) = \lambda_i e^{-\lambda_i z_i}$ for $z_i \in [0,\infty)$. Then $z_if'_i(z_i)/f_i(z_i) = -\lambda_i z_i$
	and thus $$\mathcal{X} = \left\{ z: \sum \lambda_i z_i > n+1 \right\}; \qquad \mathcal{Y} = \left\{ z: \sum \lambda_i z_i \leq n+1 \right\} \setminus \{\vec{0}\}.$$ We also have
	$$\mu_d^{\mathcal{X}}(x) = \prod \lambda_i \left(\sum \lambda_i x_i - n - 1 \right)e^{-\sum \lambda_i x_i} ; \qquad \nu_d^{\mathcal{Y}}(y) = \prod \lambda_i \left( n+1- \sum \lambda_i y_i\right)e^{-\sum \lambda_i y_i} $$
where $\mu^{\mathcal{X}}_d$ and $\nu^{\mathcal{Y}}_d$ are defined on $\mathcal{X}$ and $\mathcal{Y}$, respectively. 
We extend $\mu^{\mathcal{X}}$ to the space $\mathcal{X}_0$ by setting $\mu^{\mathcal{X}_0}(\{\vec{0}\}) = 1.$
%In Section~\ref{exponentialsolution}, we  show that $\mu^{\mathcal{X}_0}(\{\vec{0}\}) = 1$.
\end{example}

\begin{example}[Power-Law Distribution]
Suppose that the pdf of item $i$ is given by $f_i(z_i) = \frac{c_i - 1}{(1+z_i)^{c_i}}$ for $z_i \in [0,\infty)$, where each $c_i > 2$. (This restriction is necessary for $\lim_{z_i \rightarrow \infty}z_i^2f_i(z_i) = 0$.) Then $-z_if_i'(z_i)/f_i(z_i) = c_i z_i/(1+z_i)$ and thus
$$\mathcal{X} = \left\{ z : \sum \frac{c_i z_i}{1+z_i} > n+1 \right\}; \qquad \mathcal{Y} = \left\{ z : \sum \frac{c_i z_i}{1+z_i} \leq n+1 \right\} \setminus \{\vec{0}\}.$$
We also have
%\begin{align*}
%	\mu_d^{\mathcal{X}}(x) &= \left(\prod \frac{c_i - 1}{(1+x_i)^{c_i}}\right) \left(\sum \frac{c_ix_i}{1+x_i} - n - 1 \right)\\
%	 \nu_d^{\mathcal{Y}}(y) &= \left( \prod \frac{c_i - 1}{(1+y_i)^{c_i}}\right)\left(n+1 - \sum \frac{c_iy_i}{1+y_i}  \right) .
%\end{align*}
$$
	\mu_d^{\mathcal{X}}(x) = \prod_i \frac{c_i - 1}{(1+x_i)^{c_i}} \cdot \left(\sum_j \frac{c_jx_j}{1+x_j} - n - 1 \right); \;
	 \nu_d^{\mathcal{Y}}(y) = \prod_i \frac{c_i - 1}{(1+y_i)^{c_i}}\cdot \left(n+1 - \sum_j \frac{c_jy_j}{1+y_j}  \right) .
$$
Finally, we extend the measure $\mu^\mathcal{X}$ to the space $\mathcal{X}_0$ by setting $\mu^{\mathcal{X}_0}(\{\vec{0}\}) = 1$.
	%In Section~\ref{powerlaw}, we show that $\mu^{\mathcal{X}_0}(\{\vec{0}\}) = 1$.
\end{example}

%\begin{example}[Beta Distribution]
%
%Suppose that the pdf of each item $i$ is given by $f_i(x_i) = \frac{1}{B(a_i,b_i)} x_i^{a_i-1}(1-x_i)^{b_i-1}$ for $x_i \in [0,1)$, where $B(a_i,b_i)$ is the ``beta function'' and is used for normalization.
%
%We now compute
%\begin{align*}
%-\nabla f(x) \cdot x = -\sum_i x_if'_i(x_i) \prod_{j \neq i}f_j(x_j)
%\end{align*}
%
%\bigskip
%
%\bigskip
%
%\bigskip
%
%
%
%Suppose the pdf is given by $f_i(x_i) =  B(a_i,b_i) x_i^{a_i-1}(1-x_i)^{b_i-1}$ for $0 \leq x_i< 1$ and 0 for $x_i \geq 1$, where $a_i \geq 0$, $b_i > 2$ and $ B(a_i,b_i)$ is chosen so that $f_i$ is a probability distribution. We now compute, for $0 \leq x_i < 1$:
%	$$f'_i(x_i) = B(a_i,b_i) (a_i-1) x_i^{a_i - 2}(1-x_i)^{b_i-1} - B(a_i,b_i) (b_i-1)x_i^{a_i-1}(1-x_i)^{b_i-2}.$$
%	Notice that $f'_i(1) = 0$ since $b_i > 2$, and thus $f_i$ is continuously differentiable on all of $\mathbb{R}_{\geq 0}$.
%	We now compute, for $0 \leq x_i < 1$:
%	$$x_i f'_i(x_i)/f_i(x_i) = \frac{(a_i-1) (1-x_i) - (b_i-1)x_i}{(1-x_i)} = a_i -1- \frac{x_i(b_i-1)}{1-x_i} = a_i  + b_i - 2 - \frac{b_i-1}{1-x_i}, $$
%	a decreasing function. Note that, since we will be multiplying $x_i f'_i(x_i)/f_i(x_i)$ by $f_i(x_i)$ in our definitions of $\mu_d^{\mathcal{X}}$ and $\nu_d^{\mathcal{Y}}$, and since $f'_i(x_i) = 0$ for all $x_i \geq 1$, we can treat ``$x_i f'_i(x_i)/f_i(x_i)$'' as having value $0$ for all $x_i \geq 1$.
%	
%	
%	
%	
%	We therefore have $\mathcal{X} = \left\{ x: \sum \frac{b_i-1}{1-x_i} \geq + \sum a_i + \sum b_i + 1 - n \right\}$, a monotone increasing set.
%\end{example}

\section{The Relaxed Problem}\label{sec:relaxedprogram}

We define the ``cost function'' $c: \mathcal{X}_0 \times \mathcal{Y} \rightarrow \mathbb{R}$ by $$c(x,y) \triangleq \sum_i \max \{x_i - y_i, 0 \}.$$
 We notice that, if $u$ satisfies our constraint that $\nabla u \in [0,1]^n$ almost everywhere and is continuous, we have as a consequence that $u(x) - u(y) \leq c(x,y)$ for all $x \in \mathcal{X}, y \in \mathcal{Y}$.
We therefore consider the relaxed problem of maximizing $$\int_{\mathcal{X}_0} u(x) d\mu^{\mathcal{X}} - \int_{\mathcal{Y}} u(y) d\nu^{\mathcal{Y}}$$
subject only to the constraint that $u(x) - u(y) \leq c(x,y)$ for all $x \in \mathcal{X}_0$ and all $y \in \mathcal{Y}$. 
%Notice that we have relaxed the convexity constraint of the original problem. 
The optimal value of this relaxed program is therefore at least as large as the optimal revenue of the mechanism design program. We hope to identify scenarios in which we can solve the relaxed problem and in which the solution to the relaxed problem satisfies all of the original constraints. Indeed, if the relaxed problem's solution satisfies all original constraints, the solution is also optimal for the original problem.
%\footnote{Since $\mathcal{X}_0$ and $\mathcal{Y}$ have the same measure, we can add any constant to our utility function in the relaxed problem without changing its objective value. Therefore, we ignore the $u(d_-)=0$ constraint when studying the relaxed problem.}

\section{Dual Relaxed Problem}

Consider the problem
$$\inf \int_{\mathcal{X}_0 \times \mathcal{Y}} c(x,y) d\gamma(x,y) :  \gamma \in \Gamma({\mu}^{\mathcal{X}_0},\nu^{\mathcal{Y}})$$
where $\Gamma({\mu}^{\mathcal{X}_0},\nu^{\mathcal{Y}})$ is the set of all measures on $\mathcal{X}_0 \times \mathcal{Y}$ with marginal measures ${\mu}^{\mathcal{X}_0}$ and $\nu^{\mathcal{Y}}$, respectively.\footnote{That is, $\gamma(A, \mathcal{Y}) = \mu^{\mathcal{X}_0}(A) \nu^{\mathcal{Y}}(\mathcal{Y})$ %\cdot \nu^{\mathcal{Y}}(\mathcal{Y})$ 
%NOTE THAT THIS WAS CHANGED- FEB 2.
for all measurable $A \subseteq \mathcal{X}_0$, and analogously for the other marginal. } We call this the \emph{dual relaxed problem}. Informally, this problem asks for the minimum cost way of ``transporting mass'' to change the measure $\nu^\mathcal{Y}$ into the measure $\mu^{\mathcal{X}_0}$, and is a continuous analog of the minimum-weight bipartite matching problem. Analogous to Monge-Kantorovich duality from optimal transport theory \cite{Villani}, we have the following theorem, which is a continuous version of the relationship between minimum-weight bipartite matching and its linear programming dual:

\begin{theorem}\label{thm::main}
Suppose that there exist $u^*$, ${\gamma}^*$ feasible for the relaxed problem and the dual relaxed  problem, respectively, such that ${u}^*(x) - {u}^*(y) = c(x,y)$, ${\gamma}^*$-almost surely. Then ${u}^*$ is optimal for the relaxed problem and ${\gamma}^*$ is optimal for the dual relaxed problem.
\end{theorem}

\begin{proof}
We have\footnote{Recalling that $\nu^\mathcal{Y}(\mathcal{Y}) = \mu^{\mathcal{X}_0}(\mathcal{X}_0)$.} 
$$\int_{\mathcal{X}_0 \times \mathcal{Y}} c(x,y) d\gamma \geq \int_{\mathcal{X}_0 \times \mathcal{Y}}(u(x)-u(y))d\gamma= \nu^{\mathcal{Y}}(\mathcal{Y})\left( \int_{\mathcal{X}_0} u(x) d{\mu}^{\mathcal{X}_0} - \int_{\mathcal{Y}}u(y)d\nu^{\mathcal{Y}}\right)$$
for any feasible $u$ and $\gamma$. Therefore, the optimal value of the relaxed primal is at most $1/\nu(\mathcal{Y})$ times the optimal value of the relaxed dual. We also have
\begin{align*}
\int_{\mathcal{X}_0} {u}^*(x) d{\mu}^{\mathcal{X}} - \int_{\mathcal{Y}}{u}^*(y)d\nu^{\mathcal{Y}} &= \frac{1}{\nu^{\mathcal{Y}}(\mathcal{Y})}\left( \int_{\mathcal{X}_0 \times \mathcal{Y}}{u}^*(x)d{\gamma}^*  - \int_{\mathcal{X}_0 \times \mathcal{Y}} {u}^*(y)d{\gamma}^* \right)\\
&=  \frac{1}{\nu^{\mathcal{Y}}(\mathcal{Y})} \int_{\mathcal{X}_0 \times \mathcal{Y}} ({u}^*(x)-{u}^*(y)) d{\gamma}^* = \frac{1}{\nu^{\mathcal{Y}}(\mathcal{Y})} \int_{\mathcal{X}_0 \times \mathcal{Y}}c(x,y) d{\gamma}^*.
\end{align*}
Since we have found ${u}^*$ and ${\gamma}^*$ such that the value of the relaxed primal is exactly $1/\nu^{\mathcal{Y}}(\mathcal{Y})$ the value of the relaxed dual, we conclude that ${u}^*$ and ${\gamma}^*$ are both optimal for their respective problems.
\end{proof}
We may sometimes refer to $\gamma^*$ as an optimal \emph{transport map} between $\mathcal{X}_0$ and $\mathcal{Y}$, since it represents the lowest-cost method of transporting mass between the two measure spaces.

\section{Strassen's Theorem and Stochastic Dominance}

Our overall goal is to find $u^*$ and $\gamma^*$ satisfying the conditions of Theorem~\ref{thm::main}, and thereby are optimal for their respective problems. %In this section, we develop mathematical tools for proving the existence of a suitable $\gamma^*$ without needing to define it explicitly.
%\subsection{Strassen's Theorem on Stochastic Dominance}
However, it may be difficult to explicitly define an appropriate measure $\gamma^* \in \Gamma({\mu}^{\mathcal{X}_0},\nu^{\mathcal{Y}})$. Instead, we will often make use of a theorem of Strassen, which allows us to prove the existence of such $\gamma^*$ by demonstrating that one measure stochastically dominates another. (Informally, measure $\alpha$ ``stochastically dominates'' $\beta$ if $\beta$ can be transformed into $\alpha$ by moving mass only in positive directions.)

\begin{definition}
We denote by $\preceq$ the partial ordering on $\mathbb{R}^n_{\geq 0}$ where $a \preceq b$ if and only if $a_i \leq b_i$ for all $i$. In terms of this partial ordering, we make the following definitions.

A function $f: \mathbb{R}^n_{\geq 0} \rightarrow \mathbb{R}$ is \emph{increasing} if  $a \preceq b \Rightarrow f(a) \leq f(b).$

A set $S \subset \mathbb{R}^n$ is {\em increasing} if $a \in S$ and $a \preceq b$ implies $b \in S$, and {\em decreasing} if  $a \in S$ and $b \preceq a$ implies $b \in S$.
\end{definition}

\begin{definition}
For two measures $\alpha$, $\beta$ on $\mathbb{R}^n_{\geq 0}$, we say that $\alpha$ \emph{stochastically dominates} $\beta$ (with respect to the partial order $\preceq$), denoted $\beta \preceq \alpha$, if
$\int_{\mathbb{R}^n_{\geq 0}} f d\alpha \geq \int_{\mathbb{R}^n_{\geq 0}}f d\beta$
for all increasing bounded measurable functions $f$.\footnote{Throughout the paper whenever we use the terms ``measure'' or ``measurable'' we use them with respect to the Borel $\sigma$-algebra.} Similarly, if $g, h$ are density functions, $h \preceq g$ if $\int_{\vec{x} \in \mathbb{R}^n{_\geq 0}} f(\vec{x}) h(\vec{x}) d\vec{x}    \leq \int_{\vec{x} \in \mathbb{R}^n_{\geq 0}}f(\vec{x}) g(\vec{x}) d\vec{x}$ for all increasing bounded measurable functions $f$.
\end{definition}

We now apply a theorem of Strassen for the partial order $\preceq$, using the formulation from \citeN{Lindvall99} and noting that the set $\{(a,b): a \preceq b\}$ is closed:

\begin{theorem}[Strassen]\label{thm:strassen}
If $\alpha$ and $\beta$ are probability measures on $\mathbb{R}^n_{\geq 0}$ and $\alpha$ stochastically dominates $\beta$ with respect to $\preceq$, then there exists a probability measure $\hat{\gamma} \in \Gamma(\alpha, \beta)$ on $\mathbb{R}^n_{\geq 0} \times \mathbb{R}^n_{\geq 0}$ with marginals $\alpha$ and $\beta$ respectively such that $\hat{\gamma}(\{(a,b): b \preceq a\}) = 1$.
\end{theorem}

The choice of $\preceq$ is justified by the following observation (as will become clear in the next section): if types $x \in \mathcal{X}$ and $y \in \mathcal{Y}$ both receive the grand bundle at price $p$ under a utility function $u$, then $u(x)-u(y)=c(x,y)$ if and only if $y \preceq x$.

%To apply this theorem, suppose we wish to prove that a mechanism which gives the grand bundle at price $p$ to all bidders in a set $\mathcal{W}\subseteq \mathbb{R}^n_{\geq 0}$ is optimal, where $\mu^{\mathcal{X}_0}(\mathcal{W} \cap \mathcal{X}_0) = \nu^{\mathcal{Y}}(\mathcal{W} \cap \mathcal{Y})$. We prove optimality by constructing $\gamma^*$ such that $\gamma^*(\mathcal{W} \cap \mathcal{X}_0, \cdot)$ is fully supported on $\mathcal{W} \cap \mathcal{Y}$, and $\gamma^*(\cdot, \mathcal{W} \cap \mathcal{Y})$ is fully supported on $\mathcal{W} \cap \mathcal{X}_0$.  It can be difficult, however, to explicitly define $\gamma^*$ in this regime. Instead, by Strassen's theorem, it suffices to demonstrate that $\nu^{\mathcal{Y}}|_{\mathcal{Y} \cap \mathcal{W}} \preceq \mu^{\mathcal{X}_0}|_{\mathcal{X}_0 \cap \mathcal{W}}$, where 
%$\nu^{\mathcal{Y}}|_{\mathcal{Y} \cap \mathcal{W}}$ and $\mu^{\mathcal{X}_0}|_{\mathcal{X}_0 \cap \mathcal{W}}$ denote the restrictions of $\nu^{\mathcal{Y}}$ and $\mu^{\mathcal{X}_0}$ to $\mathcal{Y} \cap \mathcal{W}$ and $\mathcal{X}_0 \cap \mathcal{W}$, respectively.
%CHANGED ON FEB 2- CHANGED MATHCAL(C) TO MATHCAL(W)

\section{Optimality of Grand Bundling}\label{bundlingsection}

Our goal in this section is to identify conditions under which the solution to the relaxed problem is a take-it-or-leave-it offer of the grand bundle for some price $p^*$. In this case, the relaxed solution is clearly also a solution of the original mechanism design instance. (In particular, the utility function is convex.) Our proof of optimality relies on Theorem~\ref{thm::main}.

\begin{definition}
For any $p > 0$, we define $Z_p$ and $W_p$ by
$$Z_p \triangleq \left\{ y \in D \setminus \{d_-\} : \sum y_i \leq p \right\}; \qquad W_p \triangleq D \setminus (Z_p \cup \{d_-\}).$$

\end{definition}

%For notational convenience, we extend $\mu^{\mathcal{X}}$ and $\nu^{\mathcal{Y}}$ to all of $\mathbb{R}^2_{\geq 0}$ by setting their densities to equal to 0 off of their respective spaces.
 %By Theorem~\ref{thm::main}, we will prove optimality by duality. That is, we will map all points 
That is, $Z_p$ (along with $d_-$) is the set of types which will receive no goods under the grand bundle price $p$.
We aim to find a price $p$ such that $Z_p \subseteq \mathcal{Y}$ and such that our transport map can send all of $Z_p$ to $\{d_-\}$.
\begin{definition} \label{def: critical price}
We say that $p^* > 0$ is a \emph{critical bundle price} if $Z_{p^*} \subseteq \mathcal{Y}$ and
$$\int_{Z_{p^*}} \nu^{\mathcal{Y}}(y)dy = \int_{\mathcal{Y}} \nu^{\mathcal{Y}}(y)dy  - \int_{\mathcal{X}}\mu^{\mathcal{X}}(x)dx = 1.$$
\end{definition}
%Note that $\int_{\mathcal{Y}} \nu^{\mathcal{Y}}(y)dy  - \int_{\mathcal{X}}\mu^{\mathcal{X}}(x)dx = 1$, by Claim~\ref{equalsone}.
Even if a critical bundle price $p^*$ exists, it is not necessarily true that the optimal mechanism is a simple take-it-or-leave-it offer of all goods for price $p^*$. 

%Such a mechanism is optimal, however, will suffice to prove, by Strassen's theorem, that $\mu^{\mathcal{X}}|_{W_{p^*}} \succeq \nu^{\mathcal{Y}}|_{W_{p^*}}.$ %Indeed, in such a scenario, there must exist a suitable transport map $\gamma^*$. 

\begin{theorem}\label{bundlingthm}
Suppose that there exists a critical bundle price $p^*$ such that $\mu^{\mathcal{X}}|_{W_{p^*}} \succeq \nu^{\mathcal{Y}}|_{W_{p^*}}.$ Then the optimal mechanism is a take-it-or-leave-it offer of the grand bundle for price $p^*$.\footnote{If $\alpha$ is a measure and $S \subset \mathbb{R}^n_{\geq 0}$, then the restriction $\alpha|_S$ is the measure such that $\alpha|_S(A) \triangleq \alpha(A \cap S)$ for all measurable $A$.}
\end{theorem}

\begin{proof}
Suppose that there exists such a critical bundle price $p^*$. By Strassen's theorem, since $\mu^{\mathcal{X}}|_{W_{p^*}} \succeq \nu^{\mathcal{Y}}|_{W_{p^*}}$, there exists a transport map $\gamma_1^* \in \Gamma(\mu^{\mathcal{X}}|_{W_{p^*}}, \nu^{\mathcal{Y}}|_{W_{p^*}})$ such that, for  $x \in W_{p^*} \cap \mathcal{X}$ and $y \in W_{p^*} \cap \mathcal{Y}$, it holds that $x \succeq y$, $\gamma_1^*$ almost-surely. Since $x$ and $y$ are in $W_{p^*}$, a bidder of either type receives the grand bundle for price $p^*$. Thus, under the utility $u^*$ of the bundling mechanism:
$$u^*(x) - u^*(y) = (\sum x_i - p^*) - (\sum y_i - p^*) = \sum_i (x_i - y_i) = c(x,y),$$
where the final equality follows from $x \succeq y$.
We now extend $\gamma_1^*$ to a transport map $\gamma^* \in (\mu^{\mathcal{X}_0},\nu^\mathcal{Y})$ by mapping between $d_- \in \mathcal{X}_0$ and all of $Z_{p^*} \subseteq \mathcal{Y}$. Indeed, such a map exists since $\mu^{\mathcal{X}_0}(\{d_-\}) = \nu^\mathcal{Y}(Z_{p^*})$. We notice that any bidder of type $y \in Z_{p^*}$ or $d_-$ receives zero utility under $u^*$. Thus, we have
$u^*(d_-) - u^*(y) = 0 = c(d_-,y),$
since $d_- \preceq y$.

The existence of the transport map $\gamma^*$ proves that the bundling utility function $u^*$ is optimal for the relaxed problem, by Theorem~\ref{thm::main}. Since $u^*$ clearly satisfies all of the original mechanism design constraints, it is indeed the utility function of the optimal mechanism.
\end{proof}

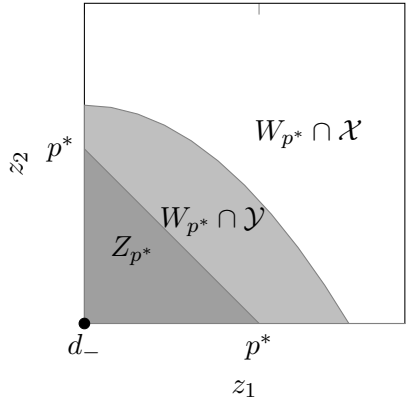
\begin{figure}
\begin{center}
\begin{tikzpicture}
\begin{axis}[width=2.3in, height=2.3in,ymin=1.1, ymax=3.3, xmin=0, xmax=1, xlabel=$z_1$, ylabel=$z_2$,  ytick pos=left, ytick={2.3}, yticklabels={$p^*$}, xtick={0, .545},xticklabels={$d_-$,$p^*$}]
  \addplot+[color=gray, fill=gray!50, domain=0:2,mark=none]
 %{1.9-1.5*x}
 {2.6-2.2*x*x}
 \closedcycle;
\addplot+[color=gray,fill=gray!75,domain=0:2,mark=none]
{2.3-2.2*x}
\closedcycle;
\addplot[color=black, mark=*] coordinates{
(0,1.1)};
\node at (axis cs:0.4,1.8){$W_{p^*} \cap \mathcal{Y}$};
\node at (axis cs:0.7,2.4){$W_{p^*} \cap \mathcal{X}$};
\node at (axis cs:0.15,1.6){$Z_{p^*}$};
\end{axis}
\end{tikzpicture}
\end{center}
\vspace{-.5cm}
\caption{The price $p^*$ is a critical bundle price if $Z_{p^*} \subseteq \mathcal{Y}$ and $\int_{Z_{p^*}} \nu^{\mathcal{Y}}(y)dy = \mu^{\mathcal{X}_0}(\{ d_- \})$. If $p^*$ is a critical bundle price and $\mu^{\mathcal{X}}|_{W_{p^*}} \succeq \nu^{\mathcal{Y}}|_{W_{p^*}}$, then a take-it-or-leave-it offer of all items for $p^*$ is the optimal mechanism. While this diagram is drawn with $n=2$, the result holds for arbitrary $n$.}
\end{figure}

It is oftentimes difficult to verify directly that $\mu^{\mathcal{X}}|_{W_{p^*}} \succeq \nu^{\mathcal{Y}}|_{W_{p^*}}$. In two dimensions, we will make use of Theorem~\ref{regionthm}, which provides a sufficient condition for a measure to stochastically dominate another. This theorem is an application of the lemma (proven in the Appendix) that an equivalent condition for dominance is that one measure has more mass than the other on all sets which are the union of \textit{finitely many} ``increasing boxes.'' Under the conditions of Theorem~\ref{regionthm}, we are able to induct on the number of boxes by removing one box at a time. For application to Theorem~\ref{bundlingthm}, we only need to use Theorem~\ref{regionthm} in the special case that $\mathcal{C}=D$ and $R = Z_{p^*}$, but in Section~\ref{exponentialsolution} we will use the more general form of the theorem. The proof of Theorem~\ref{regionthm} is in the Appendix.

Informally, Theorem~\ref{regionthm} deals with the scenario where two {density functions,} $g$ and $h$, are both nonzero only on some set $\mathcal{C} \setminus R$, where $R$ is a decreasing subset of $\mathcal{C}$. To prove that $g \succeq h$, it suffices to verify that (1) $g-h$ has an appropriate form (2) the integral of $g-h$ on $\mathcal{C}$ is \emph{positive} and (3) if we integrate $g-h$ along either a vertical or horizontal line outwards starting from any point in $R$, the result is \emph{negative}.

\begin{theorem}\label{regionthm}
Let $\mathcal{C} = [c_1, d_1^+) \times [c_2, d_2^+)$, $R$ be a decreasing nonempty subset of $\mathcal{C}$, and $g, h: \mathcal{C} \rightarrow \mathbb{R}_{\geq 0}$ be {bounded density functions which are $0$ on $R$, have finite total mass, and satisfy}
\begin{itemize}
	\item $\int_{\mathcal{C}} (g - h)dxdy \geq 0$.
	\item For any basis vector $e_i \in \{(0,1),(1,0) \}$ and any point $z^* \in R$:
	$$\int_0^{d_i^+ - z^*_i} g(z^* + \tau e_i) - h(z^* + \tau e_i)d\tau \leq 0.$$
	\item There exist non-negative functions $\alpha: [c_1, d_1^+) \rightarrow \mathbb{R}_{\geq 0}$, $\beta : [c_2, d_2^+) \rightarrow \mathbb{R}_{\geq 0}$ and an increasing function $\eta : \mathcal{C} \rightarrow \mathbb{R}$ such that
	$$g(z_1,z_2) - h(z_1,z_2) = \alpha(z_1)\cdot \beta(z_2) \cdot \eta(z_1,z_2)$$
	for all $(z_1,z_2) \in \mathcal{C} \setminus R$.
\end{itemize}
Then $g \succeq h$.
\end{theorem}
When we prove optimality of grand bundling, we will apply Theorem~\ref{regionthm} with $g$ and $h$ being the densities of $\mu^{\mathcal{X}}|_{W_{p^*}}$ and $\nu^{\mathcal{Y}}|_{W_{p^*}}$, respectively.

\section{Numerical Example: Bundling Two Power-Law Items}\label{powerlaw}
In this section, we derive, as an application of Theorem~\ref{bundlingthm}, the optimal mechanism for an instance of selling two goods which are distributed according to power-law distributions. We exhibit that the optimal mechanism is a take-it-or-leave-it offer of the grand bundle, and we find the right price to charge for it. Unlike our example in the next section for exponentially distributed items (where we perform our calculations for arbitrary pairs of exponential distributions), here we demonstrate how numerical computations can be used to prove optimality of grand bundling in a single instance. Indeed, since the integrals involved in our computations may be complicated, we suspect that this numerical approach will frequently be useful.

We proceed with our goal of deducing the optimal mechanism for selling two goods, where the distribution of each good has probability density function $f_i(z_i) = \frac{c_i - 1}{(1+z_i)^{c_i}}$ for $z_i \in [0,\infty)$. In this example, we fix $c_1 = 6$ and $c_2 = 7$.

%In this section, we verify several inequalities numerically. While short of a fully rigorous proof, all such inequalities have a large enough gap between the computed maximum value and the necessary upper bound that issues of numerical precision are not relevant.

%\subsection{Computing $\mu^{\mathcal{X}_0}(\{\vec{0}\})=\int_\mathcal{Y} \nu^\mathcal{Y}(y) dy - \int_{\mathcal{X}} \mu^{\mathcal{X}}(x) dx$}
%
%
%We begin by computing $\int_\mathcal{Y} \nu^\mathcal{Y}(y) dy - \int_{\mathcal{X}} \mu^\mathcal{X}(x) dx$, the mass to assign to $\vec{\{0\}}$ under the measure $\mu^{\mathcal{X}_0}$. This quantity is
%$$\int_0^\infty \int_0^\infty \left(3 - \frac{c_1 z_1}{1+z_1} - \frac{c_2 z_2}{1+z_2} \right)\frac{c_1 - 1}{(1+z_1)^{c_1}}\cdot \frac{c_2 - 1}{(1+z_2)^{c_2}} dz_1 dz_2,$$
%which evaluates to 1 for all $c_1, c_2 > 2$.
%
\subsection{Numerically Computing the Critical Bundle Price}
%Now fix $c_1 = 6$ and $c_2 = 7$. 
We begin by computing the critical bundle price $p^*$, as prescribed by Definition~\ref{def: critical price}. In particular, we first solve for $p^*$ such that 
$$\int_0^{p^*} \int_0^{p^*-z_2} \left(3 - \frac{c_1 z_1}{1+z_1} - \frac{c_2 z_2}{1+z_2} \right)\frac{c_1 - 1}{(1+z_1)^{c_1}}\cdot \frac{c_2 - 1}{(1+z_2)^{c_2}} dz_1 dz_2=1$$
and numerically determine that $p^* \approx .35725$. Thus, if there indeed exists a critical bundle price, it must have value $p^*$.

Next, to confirm that $p^*$ is a critical bundle price, we must verify that $Z_{p^*} \subseteq {\cal Y}$. That is, we must show that for all $z$ with $z_1 + z_2 \leq p^*$, it holds that 
$\frac{c_1z_1}{1+z_1} + \frac{c_2 z_2}{1+z_2} \leq 3.$

Since the left-hand side of the inequality is an increasing function, it suffices to prove the inequality when $z_1 + z_2 = p^*$. Substituting for $z_2$, the left-hand side of the above inequality becomes
$\frac{c_1z_1}{1+z_1} + \frac{c_2 p^* -c_2 z_1}{1+p^* - z_1}.$

We numerically compute, after setting $c_1 = 6$ and $c_2=7$, that the expression is maximized by $z_1 = 0.133226$, achieving value $1.98654$. Since $1.98654$ is significantly less than $3$, we conclude that $p^*$ is indeed a critical price, even taking into consideration possible errors of precision.

\subsection{Verifying Stochastic Dominance}
Given Theorem~\ref{bundlingthm}, all that remains to prove optimality of grand bundling is to verify that $\mu^{\mathcal{X}}|_{W_{p^*}} \succeq \nu^{\mathcal{Y}}|_{W_{p^*}}.$ We prove this stochastic dominance using Theorem~\ref{regionthm} with $g = \mu_d^\mathcal{X}|_{W_{p^*}}$ and $h = \nu_d^\mathcal{Y}|_{W_{p^*}}$. Indeed, Theorem~\ref{regionthm} applies in this case, and the verification of stochastic dominance is in Appendix~\ref{powerappendix}.
In conclusion, we obtain:

\begin{example}
The optimal mechanism for selling two independent goods with densities $f_1(z_1) = 5/(1+z_1)^6$ and $f_2(z_2) = 6/(1+z_2)^7$ respectively is a take-it-or-leave-it offer of the bundle of the two goods for price $p^* \approx .35725$.
\end{example}
%A nearly identical approach can prove that grand bundling is optimal for $c_1$ ``close'' to $c_2$.

\section{Complete Solution for Two Exponential Items}\label{exponentialsolution}
In this section we provide a closed-form description of the optimal mechanism for two independent exponentially distributed items. The description of the optimal mechanism is given by Theorem~\ref{twoitems}. In this case, the optimal mechanism has richer structure than only offering the grand bundle at some price.

In the next subsections, when it does not provide significant additional complications, we perform  computations for $n$ exponentially distributed items, as these calculations may prove useful in extensions of our result. We denote by $\lambda_i$ the parameters of the exponential distributions, and by $\lambda_{\min} = \min_i \lambda_i$ and $\lambda_{\max} = \max_i \lambda_i$

\subsection{The Critical Price $p^*$}

\begin{definition}
For any $0 < p \leq 2/\lambda_{\min}$, we define the \emph{zero space}, $Z'_p \subset \mathcal{Y} \cup \{\vec{0}\}$, to be
$$Z'_p \triangleq \left\{ y \in \mathbb{R}^2_{\geq 0} : \sum y_i \leq p \textrm{ and } \sum \lambda_i y_i \leq 2 \right\}.$$
\end{definition}
See Figure~\ref{fig: exp items first} for an example of a zero space. 
In terms of $Z'_p$ we make the following definition.
\begin{definition} \label{def: critical price 2}
For all $\lambda_1, \lambda_2 > 0$, the {\em critical price} $p^*=p^*(\lambda_1,\lambda_2)$ is the unique $0 < p^* \leq 2/\lambda_{\min}$ such that $$\nu^{\mathcal{Y}}(Z'_{p^*}) = \int_{\mathcal{Y}}d\nu^{\mathcal{Y}}(y) - \int_{\mathcal{X}} d\mu^{\mathcal{X}}(x) = 1.$$
\end{definition}
Note that the critical price differs from the critical bundle price defined in Section~\ref{bundlingsection} in that the critical price is defined with respect to the space $Z'_p$ instead of $Z_p$.

\begin{claim} \label{claim:critical price}
For all $\lambda_1, \lambda_2 > 0$, the critical price is well-defined.
\end{claim}
\begin{proof}
We need to show that there is a unique solution to the equation defining the critical price. %We have already shown that $\int_{\mathcal{Y}}d\nu^{\mathcal{Y}}(y) - \int_{\mathcal{X}} d\mu^{\mathcal{X}}(x) = 1$.

We note that $\nu^{\mathcal{Y}}(Z'_p)$ is strictly increasing in $p$ for $p$ in the range $[0,2/\lambda_{\min}]$, and that $\nu^{\mathcal{Y}}(Z_0) = 0$. Therefore, all that remains is to show that $\nu^{\mathcal{Y}}(Z'_{2/\lambda_{\min}}) \geq 1.$ We note that
$Z'_{2/\lambda_{min}} = \left\{ y : \sum \lambda_i y_i \leq 2 \right\},$
and we now compute
\begin{align*}
\int_{Z'_{2/\lambda_{\min}}}d\nu^{\mathcal{Y}}(y) &= {\lambda_1\lambda_2}\int_0^{\frac{2}{\lambda_2}}\int_0^{\frac{2}{\lambda_1}-\frac{\lambda_2 y_2}{\lambda_1}}(3-\lambda_1y_1-\lambda_2y_2)e^{-\lambda_1 y_1 - \lambda_2 y_2}dy_1dy_2\\
&= {\lambda_1 \lambda_2}\int_0^{2/\lambda_2}\frac{e^{-\lambda_2 y_2}(2-\lambda_2y_2)dy_2}{\lambda_1} = {\lambda_1\lambda_2}\cdot \frac{1+\frac{1}{e^2}}{\lambda_1\lambda_2} >
1,
\end{align*}
as desired.
\end{proof}

Our goal in the remainder of this section is to prove the following theorem:
\begin{theorem}\label{twoitems}
For all $\lambda_1 \geq \lambda_2 > 0$, the optimal utility function is the following:
$$
u_{p^*}(z_1,z_2) =
\begin{cases}
0 & \mbox{ if } z_1+z_2 \leq p^* \textrm{ and } \lambda_1z_1+\lambda_2z_2 \leq 2;\\
z_1+z_2\frac{\lambda_2}{\lambda_1} - \frac{2}{\lambda_1} & \mbox{ if }z_2\left(\frac{\lambda_1 - \lambda_2}{\lambda_1} \right) \leq p^* - \frac{2}{\lambda_1} \textrm{ and } \lambda_1z_1 + \lambda_2z_2 > 2;\\
z_1+z_2 - p^* & \mbox{otherwise,}
\end{cases}
$$where $p^*=p^*(\lambda_1,\lambda_2)$ is the critical price of Definition~\ref{def: critical price 2}.
In particular, the optimal mechanism offers the following menu:
\begin{enumerate}
\item Receive nothing, pay 0.
\item Receive the first item with probability 1 and the second item with probability $\lambda_2/\lambda_1$, pay $2/\lambda_1$.
\item Receive both items, pay $p^*$.
\end{enumerate}
\end{theorem}

%\subsection{Two Useful Features}
\subsection{The Absorption Hyperplane}
%Suppose that there exists a price $p$ such that $0 < p \leq (n+1)/\max\{\lambda_i\}$ such that $\nu^\mathcal{Y}(\{z: \sum z_i \leq p\}) = 1/\prod \lambda_i^2$. Our overall goal is to identify conditions which imply that $\nu_p \preceq \mu$.

A useful feature of independent exponential distributions is that our measures $\mu^{\mathcal{X}_0}$ and $\nu^{\mathcal{Y}}$ give rise to a set $\mathcal{H} \subset \mathbb{R}^n_{\geq 0}$, for which integrating the difference of the densities of $\mu^{\mathcal{X}_0}$ and $\nu^{\mathcal{Y}}$ outwards along any line starting from $\mathcal{H}$ yields 0. This set $\mathcal{H}$ provides useful geometric intuition behind the structure of the optimal mechanism.

\begin{claim}\label{hyperplane}
Suppose $z \in \mathbb{R}^n_{\geq 0}$ satisfies $\sum \lambda_j z_j = n$. Then, for any vector $\vec{v} \in \mathbb{R}^n_{\geq 0}$: %$i \in \{ 1,\ldots, n\}$:
$$\int_{0}^\infty \left( n+1 - \sum_{i}\lambda_i (z_i+\tau {v}_i) \right)e^{- \sum_{i}\lambda_i (z_i + \tau {v}_i)}d\tau = 0.$$
%$$\int_{z_i}^\infty \left( n+1 - \sum_{j \neq i}\lambda_j z_j - \lambda_ix_i \right)e^{-\lambda_i x_i - \sum_{j \neq i}\lambda_j z_j}dx_i = 0.$$
\end{claim}
\begin{proof}
We have
\begin{align*}
\int_{0}^\infty \left( n+1 - \sum_{i}\lambda_i (z_i+\tau v_i) \right)e^{- \sum_{i}\lambda_i (z_i + \tau v_i)}d\tau &= \frac{e^{-\sum_{i}\lambda_i (z_i + \tau v_i)}( \sum_{ i}\lambda_i( z_i+\tau v_i) - n)}{\sum_i \lambda_i v_i}\biggr|_{\tau=0}^\infty\\
&=-\frac{e^{-\sum_i \lambda_i z_i}(\sum_i \lambda_i z_i - n)}{\sum \lambda_i v_i} = 0.
\end{align*}\end{proof}

We refer to the set $\mathcal{H} = \{z \in \mathbb{R}^n_{\geq 0} : \sum z_i \lambda_i = n\}$ as the \emph{absorption hyperplane}, since integrating $(n+1-\sum \lambda_i x_i)e^{-\sum \lambda_i x_i}$ starting from any point in the set and going outwards in any positive direction gives 0.

\subsection{Proof of Optimality}
In this section, we prove Theorem~\ref{twoitems}, which fully specifies the optimal mechanism for two independent exponentially-distributed items.

By Theorem~\ref{thm::main}, we must prove that there exists $\gamma^* \in \Gamma(\mu^{\mathcal{X}_0},\nu^{\mathcal{Y}})$ such that $u_{p^*}(x) - u_{p^*}(y) = c(x,y)$, $\gamma^*$-almost surely.

Our transport map $\gamma^*$ will be decomposed into $\gamma_1 + \gamma_2 + \gamma_3$, where (see Figure~\ref{fig: exp items first} for reference)
\begin{enumerate}
\item  $\gamma_1 \in \Gamma( \mu^{\mathcal{X}_0}|_{\{\vec{0}\}}, \nu^{\mathcal{Y}}|_{ Z_{p^*}}) $
\item  $\gamma_2 \in  \Gamma(\mu^{\mathcal{X}_0}|_{\mathcal{B} \cap \mathcal{X}_0 }, \nu^{\mathcal{Y}}|_{ \mathcal{B} \cap \mathcal{Y}} )$, where
$\mathcal{B} = \left\{z : \lambda_1z+\lambda_2 z > 2 \textrm{ and } z_2\left(\frac{\lambda_1 - \lambda_2}{\lambda_1} \right) \leq p^* - \frac{2}{\lambda_1}  \right\}$
\item  $\gamma_3 \in \Gamma( \mu^{\mathcal{X}_0}|_{ \mathcal{W} \cap \mathcal{X}_0 }, \nu^{\mathcal{Y}}|_{ \mathcal{W} \cap \mathcal{Y}})$, where
$\mathcal{W} = \left\{z : z_1 + z_2 > p^* \textrm{ and } z_2\left(\frac{\lambda_1 - \lambda_2}{\lambda_1} \right) > p^* - \frac{2}{\lambda_1}  \right\}$
\end{enumerate}
such that $u_{p^*}(x) - u_{p^*}(y) = c(x, y)$ $\gamma_i$-almost surely for all $i$.
%\footnote{If $\alpha|_{\mathcal{C}}$ and $\beta|_{\mathcal{D}}$ are measures (but not necessarily probability measures) with the same total mass, then $\Gamma(\alpha|_{\mathcal{C}},\beta|_\mathcal{D})$ is the space of all measures $\gamma$ such that $\gamma(\cdot,\mathcal{D}) = \alpha|_{\mathcal{C}}(\cdot)$ and $\gamma(\mathcal{C},\cdot)=\beta|_{\mathcal{D}}(\cdot)$.} 
We note that for each $\gamma_i$, the marginal distributions of $\mu$ and $\nu$ that we are seeking to couple have the same total mass.

\begin{figure}[h!]
\parbox{.44\textwidth} {
  \begin{tikzpicture}
    \begin{axis}[width=2.7in, height=2.7in, ymin=0, ymax=4, xmin=0, xmax=4,  ylabel=$z_2$, xtick pos=left, xtick={.66666,1}, xticklabels={$\frac{2}{\lambda_1}$, $\frac{3}{\lambda_1}$}, ytick pos=left, ytick={2,3, 1.4}, yticklabels={$\frac{2}{\lambda_2}$,$\frac{3}{\lambda_2}$, $p^*$}]
      \addplot+[color=black,  fill=gray!20, mark=none, domain=0:1.5, samples=2]%area legend]
      {3 - 3*x}
      \closedcycle;
       \addplot+[color=gray, fill=gray, domain=.3:1.5, mark=none]
       {2-3*x}
       \closedcycle;
        \addplot+[color=gray, fill=gray, domain=0:.3,mark=none]
       {1.4-x}
       \closedcycle;
      \addplot+[color=black, domain=0:1,samples=2,  mark=none, dotted]
       {2 - 3*x} 
       \closedcycle;
      \addplot[color=black, mark=none] coordinates{
      (.3,1.1)
      (6,1.1)}; 
      \node at (axis cs:.2,.5) {$Z'_{p^*}$};
      \node (lab1) at (axis cs:1.2,.9){$\mathcal{B} \cap \mathcal{Y}$};
      \node (dest) at (axis cs:.6,.4){};
      \draw[->] (axis cs: 1.2,.75)--(dest);
      \node at (axis cs: 2.3,.5){$\mathcal{B} \cap \mathcal{X}$};
      \node at (axis cs:2.5,2.8){$\mathcal{W} \cap \mathcal{X}$};
      \node at (axis cs:1,2.2){$\mathcal{W} \cap \mathcal{Y}$};
      \node (lab2) at (axis cs:1.08,2.1){};
      \node (dest2) at (axis cs:0,1.35){};
      \draw[->] (lab2)--(dest2);
      \node (dest3) at (axis cs:.3, 1.3){};
      \node (lab3) at (axis cs:1.08,2.13){};
      \draw[->] (lab3)--(dest3);
      \node (dest4) at (axis cs: .017, 1.68){};
      \node (lab4) at (axis cs: 1,3.2){};
      \draw [->] (lab4)--(dest4);
      \node at (axis cs: 1.074,3.47){\small{absorption}};
      \node at (axis cs: 1.1,3.25){\small{hyperplane}};
    \end{axis}
  \end{tikzpicture} 
  \caption{The decomposition of $\mathbb{R}^n_{\geq 0}$ for the proof of Theorem~\ref{twoitems}. In this diagram, $p^* > 2/\lambda_1$. If $p^* \leq 2/\lambda_1$, $\mathcal{B}$ is empty.}
  \label{fig: exp items first}
}
\qquad
\begin{minipage}{.44\textwidth}
\vspace{-1cm}
%\centering
\begin{tikzpicture}
\begin{axis}[width=2.6in,height=2.6in,ymin=1.1, ymax=3.1, xmin=0, xmax=2, xlabel=$z_1$,  ytick pos=left, ytick={2,3, 1.4,1.1}, yticklabels={$\frac{2}{\lambda_2}$,$\frac{3}{\lambda_2}$, $p^*$,$ \frac{\lambda_1 p^* - 2}{\lambda_1 - \lambda_2}$}, xtick={5},xticklabels={1}]
\addplot+[color=black,  fill=gray!20, mark=none, domain=0:1.5, samples=2]%area legend]
{3 - 3*x}
\closedcycle;
%\addlegendentry{$\mathcal{Y}$}
 \addplot+[color=gray, fill=gray, domain=.3:1.5, mark=none]
 {2-3*x}
 \closedcycle;
% \addlegendentry{$Z_{p^*}$}
  \addplot+[color=gray, fill=gray, domain=0:.3,mark=none]
 {1.4-x}
 \closedcycle;
\addplot+[color=black, domain=0:1,samples=2,  mark=none, dotted]
%pattern=north east lines wide]
 {2 - 3*x} 
 \closedcycle;
\addplot[color=black, mark=none] coordinates{
(.3,1.1)
(6,1.1)}; 
\node at (axis cs:.25,1.5){$\nu^{\mathcal{Y}}|_{\mathcal{W} \cap \mathcal{Y}}$};
\node at (axis cs:1,2.2){$\mu^{\mathcal{X}}|_{\mathcal{W} \cap \mathcal{X}}$};
\end{axis}
\end{tikzpicture}
\caption{To prove the existence of $\gamma_3$, we must show that $ \mu^{\mathcal{X}}|_{ \mathcal{W} \cap \mathcal{X}} \succeq \nu^{\mathcal{Y}}|_{ \mathcal{W} \cap \mathcal{Y}}$.}
\label{case3fig}
\end{minipage}
\end{figure}

We proceed to prove the existence of each $\gamma_i$ separately.
\begin{enumerate}
\item We have $\mu^{\mathcal{X}_0}(\{\vec{0}\}) = \nu^{\mathcal{Y}}(Z'_{p^*})$, by definition of $p^*$. Furthermore, $c(\vec{0},y) = 0$ for all $y$. Therefore, since $u_{p^*}(y) = 0$ for all $y \in Z_{p^*}$, the equality $u(\vec{0}) - u(y) = c(x,y)$ is trivially satisfied for all $y \in Z'_{p^*}$.  We can therefore  take $\gamma_1$ to assign probability mass to $(\{\vec{0}\} \times Z)$ equal to $ \nu^{\mathcal{Y}}(Z)$ for each subset $Z \subseteq Z'_{p^*}$.
\item We note that $\mathcal{B}$ consists of all points to the right of the absorption hyperplane with $z_2$ coordinate less than a particular threshold. Therefore, for any $z_2^*$, we have by Claim~\ref{hyperplane}:
$$\int_\frac{2 - \lambda_2 z_2^*}{\lambda_1}^\frac{3 - \lambda_2 z_z^*}{\lambda_1} \left(3- \lambda_1 z_1 - \lambda_2 z_2^*  \right)e^{-\lambda_1 z_1 - \lambda_2 z_2^*}dz_1 =\int_\frac{3 - \lambda_2 z_2^*}{\lambda_1}^\infty \left( \lambda_1 z_1 + \lambda_2 z_2^* - 3 \right)e^{-\lambda_1 z_1 - \lambda_2 z_2^*}dz_1.$$
From this we deduce that we can choose the measure $\gamma_2 \in  \Gamma(\mu^{\mathcal{X}_0}|_{\mathcal{B} \cap \mathcal{X}_0 }, \nu^{\mathcal{Y}}|_{ \mathcal{B} \cap \mathcal{Y}} )$ so that positive density is only placed on {pairs of} points $(x,y)$ with $x_2=y_2$, i.e. with their second coordinates equal. Indeed, we notice that, for $x \in \mathcal{B} \cap \mathcal{X}_0$ and $y \in \mathcal{B} \cap \mathcal{Y}$,
$$u_{p^*}(x) - u_{p^*}(y) = x_1 + x_2 \frac{\lambda_2}{\lambda_1} - y_1 - y_2 \frac{\lambda_2}{\lambda_1}.$$
Therefore, if $x_2=y_2$ (which we take to hold $\gamma_2$-almost surely), then $u_{p^*}(x)-u_{p^*}(y) = x_1 - y_1 = c(x,y)$, as desired.
\item\label{expbundcase} In region $\mathcal{W}$, our mechanism sells the grand bundle for price $p^*$. To prove the existence of measure $\gamma_3$, it suffices to prove that $ \mu^{\mathcal{X}_0}|_{ \mathcal{W} \cap \mathcal{X}_0 } \succeq \nu^{\mathcal{Y}}|_{ \mathcal{W} \cap \mathcal{Y}}$, as illustrated in Figure~\ref{case3fig}. Indeed, then Strassen's theorem (Theorem~\ref{thm:strassen}) implies that $\gamma_3 \in \Gamma( \mu^{\mathcal{X}_0}|_{ \mathcal{W} \cap \mathcal{X}_0 }, \nu^{\mathcal{Y}}|_{ \mathcal{W} \cap \mathcal{Y}})$ exists so that pairs of points $(x,y)$ sampled from $\gamma_3$ satisfy $y \preceq x$ almost surely, which in turn implies that $u_{p^*}(x)-u_{p^*}(y) = (x_1+x_2 - p^*)-(y_1+y_2 - p^*)=c(x,y)$ almost surely.
%
%
%
%\begin{center}
%\begin{tikzpicture}
%\begin{axis}[ymin=1.1, ymax=3.1, xmin=0, xmax=2, xlabel=$z_1$, ylabel=$z_2$,  ytick pos=left, ytick={2,3, 1.4,1.1}, yticklabels={$\frac{2}{\lambda_2}$,$\frac{3}{\lambda_2}$, $p^*$,$ \frac{\lambda_1 p^* - 2}{\lambda_1 - \lambda_2}$}, xtick={5},xticklabels={1}]
%\addplot+[color=black,  fill=gray!20, mark=none, domain=0:1.5, samples=2]%area legend]
%{3 - 3*x}
%\closedcycle;
%%\addlegendentry{$\mathcal{Y}$}
% \addplot+[color=gray, fill=gray, domain=.3:1.5, mark=none]
% {2-3*x}
% \closedcycle;
%% \addlegendentry{$Z_{p^*}$}
%  \addplot+[color=gray, fill=gray, domain=0:.3,mark=none]
% {1.4-x}
% \closedcycle;
%\addplot+[color=black, domain=0:1,samples=2,  mark=none, dotted]
%%pattern=north east lines wide]
% {2 - 3*x} 
% \closedcycle;
%\addplot[color=black, mark=none] coordinates{
%(.3,1.1)
%(6,1.1)}; 
%\node at (axis cs:.25,1.5){$\nu^{\mathcal{Y}}|_{\mathcal{W} \cap \mathcal{Y}}$};
%\node at (axis cs:1,2.2){$\mu^{\mathcal{X}_0}|_{\mathcal{W} \cap \mathcal{X}}$};
%\end{axis}
%\end{tikzpicture}
%\end{center}

The desired stochastic dominance follows from Theorem~\ref{regionthm}, taking $g$ and $h$ to be the density functions of $\mu^{\mathcal{X}_0}|_{\mathcal{W} \cap \mathcal{X}}$ and $\nu^{\mathcal{Y}}|_{\mathcal{W} \cap \mathcal{Y}}$, respectively, and noticing that, restricted within ${\cal W}$, $Z'_{p^*}$ lies below the absorption hyperplane.\footnote{It is straightforward to verify that all of the conditions of Theorem~\ref{regionthm} hold. In particular, since $Z_{p^*}'$ lies under the absorption hyperplane, the second criterion for Theorem~\ref{regionthm} is trivially satisfied.} 

%Thus, by Strassen's theorem, there exists a transport map $\gamma_3 \in \Gamma( \mu^{\mathcal{X}_0}|_{ \mathcal{W} \cap \mathcal{X}_0 }, \nu^{\mathcal{Y}}|_{ \mathcal{W} \cap \mathcal{Y}})$ such that $u_{p^*}(x) - u_{p^*}(y) = c(x,y)$, $\gamma_3$-almost surely, as desired.

\end{enumerate}
%\end{proof}

This concludes the proof of Theorem~\ref{twoitems}. We notice that, if $p^* \leq 2/\lambda_{\max}$, then the region $\mathcal{B}$ is empty, and $Z'_{p^*}$ is simply the region below the $45^\circ$ line given by $z_1 + z_2 = p^*$.

\begin{corollary}
If $p^* \leq 2/\lambda_{\max}$, then the optimal mechanism is a take-it-or-leave-it offer of the grand bundle for price $p^*$.
\end{corollary}

\section{General Characterization of Two-Item Optimal Mechanisms}

We now generalize the approach of Section~\ref{exponentialsolution} to further understand the structure of optimal mechanisms. The following definition is summarized by Figure~\ref{canonicalfig}.

\begin{definition}\label{def:canonical1}
A \emph{canonical zero set} for the two-item optimal mechanism design problem is a nonempty  closed subset $Z$ of $D_1 \times D_2$, where $Z$ is decreasing and convex. We denote by $s : [d_1^-, c] \rightarrow D_2$ (with $c < d_1^+$) the outer boundary of $Z$.\footnote{While $s$ need not be a function, it is notationally convenient to refer to it as such. In Theorem~\ref{generalthm}, we only refer to ``$s(z_1)$'' when the slope of $s$ is between horizontal and $45^\circ$ downwards, and only refer to ``$s^{-1}(z_2)$'' when the slope is between $45^\circ$ downwards and vertical. To be strictly formal, we could define $s$ to be the set of points on the boundary.} That is: 
	$$Z = \left\{(z_1,z_2) : z_1 \in [d_1^-,c] \textrm{ and } z_2 \leq s(z_1) \right\}.$$
We require that $s$ be differentiable almost everywhere.

We denote by $a, b \in [d_1^-,c]$ points such that:
\begin{itemize}
	\item $0 \geq s'(z_1) \geq -1$ for $z_1 \in [d_1^-,a]$
	\item $s'(z_1) = -1$ for $z_1 \in [a,b]$
	\item $s'(z_1) \leq -1$ for $z_1 \in [b,c]$.
\end{itemize}

A canonical zero set $Z$ gives rise to a \emph{canonical partition} of $D_1 \times D_2$ into four regions, $Z$, $\mathcal{A}$, $\mathcal{B}$, $\mathcal{W}$, where:
\begin{itemize}
	\item $\mathcal{A} = ([d_1^-,a] \times D_2) \setminus Z$
	\item $\mathcal{B} = ([b,d_1^+) \times [d_2^-,s(b)]) \setminus Z$
	\item $\mathcal{W} = D_1 \times D_2 \setminus (Z \cup \mathcal{A} \cup \mathcal{B}) =  ((a,d_1^+) \times (s(b),d_2^+)) \setminus Z$,
\end{itemize}
as shown in Figure~\ref{canonicalfig}.
\end{definition}

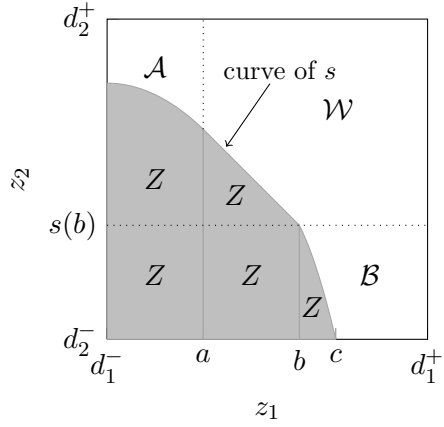
\begin{figure}\begin{center}
\begin{tikzpicture}
\begin{axis}[width=2.3in, height=2.3in,ymin=0, ymax=10, xmin=0, xmax=10, xlabel=$z_1$, ylabel=$z_2$,  ytick pos=left, ytick={3.56,0,10}, yticklabels={$s(b)$,$d_2^-$,$d_2^+$}, xtick={3, 6,7.135,0,10},xticklabels={$a$,$b$,$c$,$d_1^-$,$d_1^+$}, xtick pos=left]
  \addplot+[color=gray!70, fill=gray!50, domain=0:3,mark=none]
 {8-.16*x*x}
 \closedcycle;
   \addplot+[color=gray!70, fill=gray!50, domain=3:6,mark=none]
 {9.56-x}
 \closedcycle;
\addplot+[color=gray!70,fill=gray!50,domain=6:7.135,mark=none]
{4.56-(x-5)*(x-5)}
\closedcycle;
 \addplot+[color=black, domain=0:10,samples=2,  mark=none, dotted]
 {3.56} 
 ;
 \addplot[color=black, , dotted, mark=none] coordinates{
(3,6.56)
(3,11)
};
\node at (axis cs:1.5,5){$Z$};
\node at (axis cs:4.5,2){$Z$};
\node at (axis cs:6.4,1){$Z$};
\node at (axis cs:1.5,2){$Z$};
\node at (axis cs:4,4.5){$Z$};
\node at (axis cs:8.2,2){$\mathcal{B}$};    
\node at (axis cs:7.2,7.2){$\mathcal{W}$};
\node at (axis cs:1.5,8.5){$\mathcal{A}$};
\node (dest4) at (axis cs: 3.5, 5.7){};
\node (lab4) at (axis cs: 5.3,8.3){};
\draw [->] (lab4)--(dest4);
\node at (axis cs: 5.4,8.4){\small{curve of $s$}};
\end{axis}
\end{tikzpicture}
\end{center}
\vspace{-.5cm}
\caption{A canonical partition of $D_1 \times D_2$}\label{canonicalfig}
\end{figure}

If a canonical partition is well-formed according to the following definition, then Theorem~\ref{generalthm} characterizes the structure of the optimal mechanism.
\begin{definition}\label{def:canonical2}
A canonical partition $s, Z, \mathcal{A}, \mathcal{B}, \mathcal{W}$ is \emph{well-formed} if the following conditions are satisfied:
\begin{itemize}
	\item For all $z \in Z \setminus \{d_-\}$, it holds that $(n+1)f(z) + \nabla f(z) \cdot z \geq 0;$ i.e. $Z\setminus \{d_-\}$ lies in ${\cal Y}$.
	\item The following equality holds: $$\int_Z ((n+1)f(z) + \nabla f(z) \cdot z)dz = \int_{D_1 \times D_2} ( (n+1)f(z) + \nabla f(z) \cdot z)dz;$$
	i.e. the mass assigned by $\nu^{\cal Y}$ to $Z$ equals $\mu^{{\cal X}_0}(\{d_{-}\})$.
	\item For all $z_1 \in [d_1^-,a]$, it holds that:
	$\int_{s(z_1)}^{d_2^+} \left( (n+1)f(z) + \nabla f(z) \cdot z   \right)   dz_2 = 0;$
i.e. $\nu^{\cal Y}$ and $\mu^{\cal X}$ place the exact same mass on every vertical line originating from a point $(z_1,s(z_1))$, $z_1 \in [d_1^-,a]$ and going upwards.
	\item For all $z_2 \in [d_2^-,s(b)]$, it holds that:
	$\int_{s^{-1}(z_2)}^{d_1^+} \left((n+1)f(z) + \nabla f(z) \cdot z  \right)   dz_1 = 0;$ i.e. $\nu^{\cal Y}$ and $\mu^{\cal X}$ place the exact same mass on every horizontal line that starts from a point $(s^{-1}(z_2),z_2), z_2 \in [d_2^-,s(b)]$ and going rightwards.

	\item For all increasing subsets $T \subseteq \mathcal{W}$, it holds that:
	$\int_T ((n+1)f(z) + \nabla f(z) \cdot z)dz \leq 0;$ i.e. $\mu^{\cal X}|_{\cal W}$ stochastically dominates $\nu^{\cal Y}|_{\cal W}$.
\end{itemize}
\end{definition}

\begin{theorem}\label{generalthm}
Let $s, Z, \mathcal{A}, \mathcal{B}, \mathcal{W}$ be a well-formed canonical partition of $D_1 \times D_2$. Then
the optimal mechanism behaves as follows for a bidder of declared type $(z_1,z_2)$:
\begin{itemize}
	\item If $(z_1,z_2) \in Z$, the bidder receives no goods and is charged 0.
	\item If $(z_1,z_2) \in \mathcal{A}$, the bidder receives item 1 with probability $-s'(z_1)$, item 2 with probability 1, and is charged $s(z_1) - z_1s'(z_1)$.
	\item If $(z_1, z_2) \in \mathcal{B}$, the bidder receives item 1 with probability 1, item 2 with probability $-1/s'(s^{-1}(z_2))$, and is charged $s^{-1}(z_2) - z_2/s'(s^{-1}(z_2))$.
	\item If $(z_1,z_2) \in \mathcal{W}$, the bidder receives both goods with probability 1 and is charged $a + s(a)$ (where $a$ is as specified in Definition~\ref{def:canonical1}).
\end{itemize}
\end{theorem}
Note that Theorem~\ref{generalthm} is symmetric with respect to relabeling $z_1$ and $z_2$ and replacing $s, a, b$, and $c$ with $s^{-1}$, $s^{-1}(b)$, $s^{-1}(a)$, and $s^{-1}(0)$, respectively. Furthermore, we observe that Theorem~\ref{exponentialsolution} is a special case of this theorem, with $Z=Z'_{p^*}$ and $\mathcal{A}$ being empty.

\begin{proof}
%We note that the first condition of the partition being well-formed is equivalent to $Z \setminus \{d_-\} $ being a subset of $\mathcal{Y}$, the second condition is equivalent to $\nu^{\mathcal{Y}}(Z) =  \mu^{\mathcal{X}_0}(\{ d_-\})$, and the last condition is equivalent to $\mu^{\mathcal{X}}|_{\mathcal{W}} \succeq \nu^{\mathcal{Y}}|_{\mathcal{W}}$.
%
The utility function $u$ induced by the mechanism is as follows:
\begin{itemize}
	\item If $(z_1, z_2) \in Z$, the utility is 0.
	\item If $(z_1, z_2) \in \mathcal{A}$, the utility is $z_2 - s(z_1)$.
	\item If $(z_1, z_2) \in \mathcal{B}$, the utility is $z_1 - s^{-1}(z_2)$.
	\item If $(z_1, z_2) \in \mathcal{W}$, the utility is $z_1 + z_2 - (a+s(a))$.
\end{itemize}

It is straightforward to show that there exists a transport map $\gamma$ dual to $u$, and therefore $u$ is optimal for the relaxed problem. Indeed, $\gamma$ maps between $\{ d_- \}$ and $Z$, between $\mathcal{A} \cap \mathcal{X}$ and $\mathcal{A} \cap \mathcal{Y}$, between $\mathcal{B} \cap \mathcal{X}$ and $\mathcal{B} \cap \mathcal{Y}$, and between $\mathcal{W} \cap \mathcal{X}$ and $\mathcal{W} \cap \mathcal{Y}$. The full argument for the existence of $\gamma$ is nearly identical to the argument given in the proof of Theorem~\ref{exponentialsolution}. In particular, the existence of an appropriate map between $\mathcal{W} \cap \mathcal{X}$ and $\mathcal{W} \cap \mathcal{Y}$ follows from Strassen's theorem.% and Theorem~\ref{regionthm}.

Since a solution to the relaxed problem is not necessarily a solution to the original mechanism design instance, it remains to show that the mechanism is truthful (or, equivalently, that $u$ is convex). We consider a bidder of type $(z_1,z_2)$, and let $(z_1^*, z_2^*)$ be any other type. It is straightforward to prove, through a small amount of casework, that that the bidder's utility never increases by declaring $(z_1^*,z_2^*)$ instead of $(z_1,z_2)$. The full proof of this fact is in Appendix \ref{generalappendix}.\end{proof}
%Thus the mechanism in Theorem~\ref{generalthm} is truthful, as desired.

\section{Numerical Example: Optimal Mechanism for Two Beta Distributions}\label{sec:beta}

We  obtain a closed-form description of the optimal mechanism for two items distributed according to the Beta distributions shown below. Our approach here illustrates a general recipe for employing our characterization theorem (Theorem~\ref{generalthm}) to find closed-form descriptions of optimal mechanisms, comprising the following steps: (i) definition of the sets $S_{\rm top}$ and $S_{\rm right}$, (ii) computation of a critical price $p^*$; (iii) definition of a canonical partition in terms of (i) and (ii); and (iv) application of Theorem~\ref{generalthm}.

Suppose that the probability density functions of our items are:
$$
f_1(z_1) = \frac{1}{B(3,3)}z_1^2(1-z_1)^2; \qquad
f_2(z_2) = \frac{1}{B(3,4)}z_2^2(1-z_2)^3
$$
for all $z_i \in [0,1)$, where $B(\cdot,\cdot)$ is the ``beta function''  and is used for normalization.

We compute $-\nabla f(z) \cdot z - 3 f(z) = f_1(z_1) f_2(z_2) \left(\frac{2}{1-z_1} + \frac{3}{1-z_2}  - 12 \right).$ Thus, we define
%\begin{align*}
%\mathcal{X} &= \left\{ z\in [0,1)^2 :\frac{2}{1-z_1} + \frac{3}{1-z_2} > 12\right\}\\
% \mathcal{Y} &= \left\{ z \in [0,1)^2 :\frac{2}{1-z_1} + \frac{3}{1-z_2} \leq 12 \textrm{ and } z \neq \vec{0} \right\}
% \end{align*}
$$
\mathcal{X} = \left\{ z\in [0,1)^2 :\frac{2}{1-z_1} + \frac{3}{1-z_2} > 12\right\}; \;
 \mathcal{Y} = \left\{ z \in [0,1)^2 \setminus \{\vec{0}\} :\frac{2}{1-z_1} + \frac{3}{1-z_2} \leq 12  \right\}
$$
 and we define the densities
%\begin{align*}
%\mu_d^{\mathcal{X}}(z) &=  f_1(z_1) f_2(z_2) \left(\frac{2}{1-z_1} + \frac{3}{1-z_2}  - 12 \right) \cdot 1_{z \in \mathcal{X}}\\
%\nu_d^{\mathcal{Y}}(z) &= f_1(z_1)f_2(z_2)\left( 12 - \frac{2}{1-z_1} - \frac{3}{1-z_2} \right) \cdot 1_{z \in \mathcal{Y}}.\\
%\end{align*}
$$
\mu_d^{\mathcal{X}}(z) =  f(z) \left(\frac{2}{1-z_1} + \frac{3}{1-z_2}  - 12 \right) \cdot 1_{z \in \mathcal{X}}; \;
\nu_d^{\mathcal{Y}}(z) = f(z) \left( 12 - \frac{2}{1-z_1} - \frac{3}{1-z_2} \right) \cdot 1_{z \in \mathcal{Y}}.
$$

\noindent \textbf{Step (i).} We now define the set $S_{\textrm{top}} \subset [0,1)^2$ by the rule that $(z_1,z_2) \in S_{\textrm{top}}$ if
$$ \int_{z_2}^1 \left(-\nabla f(z_1,t) \cdot (z_1,t) - (n+1)f(z_1,t)\right)dt = 0.$$
That is, starting from any point in $z\in S_{\textrm{top}}$ and integrating $\mu_d^{\mathcal{X}}(z_1,t) - \mu_d^\mathcal{Y}(z_1,t)$ ``upwards'' from $t=z_2$ to $t=1$ yields zero.
Similarly, we say that $(z_1,z_2) \in S_{\textrm{right}}$ if
$$ \int_{z_1}^1 \left(-\nabla f(t,z_2) \cdot (t,z_2) - (n+1)f(t,z_2)\right)dt = 0.$$

Notice that, since $\mathcal{X}$ is an increasing set,  $S_{\textrm{top}}$ and $S_{\textrm{right}}$ must both be subsets of $\mathcal{Y}$. We compute analytically that $(z_1, z_2) \in [0,1)^2 $ is in  $S_{\textrm{top}}$ if and only if
$$z_1 = \frac{2(-1-3z_2-6z_2^2+25z_2^3)}{3(-1-3z_2-6z_2^2+20z_2^3)}.$$
Similarly, $(z_1,z_2) \in [0,1)^2$ is in $S_{\textrm{right}}$ if and only if
$z_2 = \frac{2(-2-4z_1-6z_1^2+27z_1^3)}{-7-14z_1-21z_1^2+72z_1^3}.$

In particular, for any $z_1 \in [0,.63718)$ there exists a $z_2$ such that $(z_1,z_2) \in   S_{\textrm{right}}$, and there does not exist such a $z_2$ if $z_1 > .63718$. Furthermore, it is straightforward to verify (by computing second derivatives in the appropriate regime) that the region below $S_{top}$ and the region below $S_{right}$ are convex.

\medskip
\noindent \textbf{Step (ii).} We now compute $p^* = 0.71307$  (this choice will be explained later- it is the $z_2$-intercept of the $45^\circ$ line in Figure~\ref{betafig} which causes $\nu^\mathcal{Y}(Z) = \mu^{\mathcal{X}_0}(\{\vec{0}\})$) and define the set $L = \left\{ z \in [0,1)^2 : z_1 + z_2 = p^* \right\}$. We compute that $L \cap S_{\textrm{top}}$ contains the point $(.16016, .55291)$ and that $L \cap S_{\textrm{right}}$ contains the point $(.62307, 0.09 )$. 
We now define the curve $s: [0,.63718] \rightarrow D_2$ by
$$
s(z_1) =
\begin{cases}
z_2 \mbox{ such that } (z_1,z_2) \in S_{\textrm{top}} & \mbox{ if } 0 \leq z_1 \leq .16016\\
.71307 - z_1 & \mbox{ if } .16016 \leq z_1 \leq .62307\\
z_2 \mbox{ such that } (z_1,z_2) \in S_{\textrm{right}} & \mbox{ if } .62307 \leq z_1 \leq .63718.\\
\end{cases}
$$
It is straightforward to verify that $s$ is a concave, decreasing, continuous function.

\medskip
\noindent \textbf{Step (iii).}
We decompose $[0,1)^2$ into the following regions:

\begin{center}
\begin{tabular}{l l}
$Z =  \left\{z : z_1 \leq 0.63718 \textrm{ and } z_2 \leq s(z_1) \right\}$; & $\mathcal{A} =  ([0,0.16016] \times (0,1)) \setminus Z$ \\
$\mathcal{B} = ((0,1) \times [0,0.09] \setminus Z$; & $\mathcal{W} =  [0,1)^2 \setminus \left( Z \cup \mathcal{A} \cup \mathcal{B}  \right)$
\end{tabular}
\end{center}
%
%\begin{align*}
%Z &= \left\{z : z_1 \leq 0.63718 \textrm{ and } z_2 \leq s(z_1) \right\}\\
%\mathcal{A} &= ([0,0.16016] \times (0,1)) \setminus Z\\
%\mathcal{B} &= ((0,1) \times [0,0.09] \setminus Z\\
%\mathcal{W} &= [0,1)^2 \setminus \left( Z \cup \mathcal{A} \cup \mathcal{B}  \right),
%\end{align*}
as illustrated in Figure~\ref{betafig}.
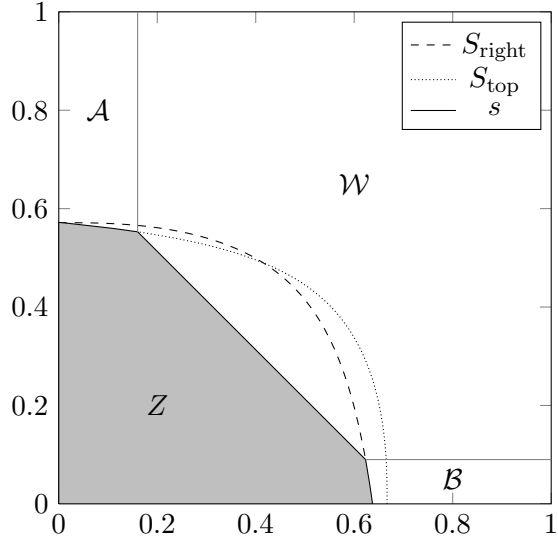
\begin{figure}
\begin{center}
\begin{tikzpicture}
%\begin{axis}[height=4in, width=4in, ymin=0, ymax=1, xmin=0, xmax=1, extra x ticks={.16016,.62307}, extra x tick labels={$.16016$, $0.623$},   xticklabel style={rotate=90}, extra y ticks={0.09}, extra y tick labels={$0.09$}, extra x tick style = {tick pos = right, ticklabel pos = upper}]
\begin{axis}[height=3.2in, width=3.2in, ymin=0, ymax=1, xmin=0, xmax=1]
%\node at (axis cs:.1,2.2){$z^{(k+1)}$};
%\node at (axis cs:.3,1.5){$z^{(k)}$};
%\node at (axis cs:.5,1.3){$z^{(k-1)}$};
%\node at (axis cs:0.15,1.4){$R$};
%Sright
\addplot[color=black, mark=none, style=dashed] coordinates{
(0,0.571428571)
(0.01,0.571426771)
(0.02,0.571414458)
(0.03,0.571381893)
(0.04,0.571320144)
(0.05,0.571221024)
(0.06,0.571077036)
(0.07,0.570881303)
(0.08,0.570627497)
(0.09,0.570309773)
(0.1,0.569922699)
(0.11,0.569461185)
(0.12,0.568920423)
(0.13,0.568295819)
(0.14,0.567582934)
(0.15,0.566777426)
(0.16,0.565874992)
(0.17,0.564871313)
(0.18,0.563762001)
(0.19,0.562542547)
(0.2,0.561208267)
(0.21,0.55975425)
(0.22,0.558175304)
(0.23,0.556465902)
(0.24,0.554620119)
(0.25,0.552631579)
(0.26,0.550493379)
(0.27,0.548198027)
(0.28,0.545737355)
(0.29,0.543102442)
(0.3,0.54028351)
(0.31,0.537269819)
(0.32,0.534049548)
(0.33,0.530609652)
(0.34,0.526935707)
(0.35,0.523011725)
(0.36,0.51881995)
(0.37,0.514340609)
(0.38,0.509551628)
(0.39,0.504428307)
(0.4,0.498942918)
(0.41,0.493064249)
(0.42,0.486757053)
(0.43,0.479981386)
(0.44,0.472691815)
(0.45,0.464836465)
(0.46,0.456355844)
(0.47,0.447181422)
(0.48,0.437233856)
(0.49,0.426420801)
(0.5,0.414634146)
(0.51,0.401746533)
(0.52,0.387606888)
(0.53,0.372034664)
(0.54,0.354812313)
(0.55,0.335675318)
(0.56,0.314298833)
(0.57,0.290279471)
(0.58,0.263110078)
(0.59,0.232144114)
(0.6,0.196544276)
(0.61,0.155206584)
(0.62,0.106645083)
(0.623068,0.090001962)
};

%Stop
\addplot[color=black, mark=none, style=densely dotted] coordinates{
(0.666666667,0)
(0.666663432,0.01)
(0.666641562,0.02)
(0.666584464,0.03)
(0.666477595,0.04)
(0.666308244,0.05)
(0.666065304,0.06)
(0.665739043,0.07)
(0.665320885,0.08)
(0.664803198,0.09)
(0.664179104,0.1)
(0.663442298,0.11)
(0.662586884,0.12)
(0.661607229,0.13)
(0.660497824,0.14)
(0.659253158,0.15)
(0.657867602,0.16)
(0.656335297,0.17)
(0.654650051,0.18)
(0.652805229,0.19)
(0.650793651,0.2)
(0.648607487,0.21)
(0.646238143,0.22)
(0.643676141,0.23)
(0.640910988,0.24)
(0.637931034,0.25)
(0.634723313,0.26)
(0.631273354,0.27)
(0.627564979,0.28)
(0.623580062,0.29)
(0.619298246,0.3)
(0.614696616,0.31)
(0.609749319,0.32)
(0.604427102,0.33)
(0.598696771,0.34)
(0.592520536,0.35)
(0.585855227,0.36)
(0.578651334,0.37)
(0.570851843,0.38)
(0.562390792,0.39)
(0.553191489,0.4)
(0.543164282,0.41)
(0.532203734,0.42)
(0.520185048,0.43)
(0.506959443,0.44)
(0.492348159,0.45)
(0.476134549,0.46)
(0.458053536,0.47)
(0.437777336,0.48)
(0.414895803,0.49)
(0.388888889,0.5)
(0.359087262,0.51)
(0.324614666,0.52)
(0.284301337,0.53)
(0.236549976,0.54)
(0.179120879,0.55)
(0.160160357,0.55291)
};

%mech
\addplot[color=black, fill=gray!50, mark=none] coordinates{
(0,0.57202)
(0.000265493,0.572)
(0.010579973,0.571)
(0.020634057,0.57)
(0.108773,0.56)
(0.16016,0.55291)
(0.160160357,0.55291)
(0.623068,0.090002)
(0.631,0.04241542)
(0.632,0.035895589)
(0.633,0.029247936)
(0.634,0.022468693)
(0.635,0.015553939)
(0.636,0.008499596)
(0.637,0.001301422)
(0.63717,0)
}\closedcycle;
\legend{$S_{\textrm{right}}$, $S_{\textrm{top}}$, $s$};
\addplot[color=gray, mark=none]coordinates{
(.62307,.09)
(1,.09)
};
\addplot[color=gray, mark=none]coordinates{
(.16016,.55291)
(.16016,1)
};
\node at (axis cs:0.08,0.8){$\mathcal{A}$};
\node at (axis cs:0.8,0.05){$\mathcal{B}$};
\node at (axis cs:0.6,0.65){$\mathcal{W}$};
\node at (axis cs:0.2,0.2){$Z$};
\end{axis}
\end{tikzpicture}
\caption{The well-formed canonical partition for $ f_1(z_1) = \frac{z_1^2(1-z_1)^2}{B(3,3)}$ and $f_2(z_2) = \frac{z_2^2(1-z_2)^3}{B(3,4)}$.}\label{betafig}
\end{center}
\end{figure}

\medskip
\noindent \textbf{Step (iv).}
We note that every point in $Z$  either lies below some point in $S_{\textrm{top}}$ or to the left of some point  in $S_{\textrm{right}}$. Thus, since $\mathcal{Y} \cup \{\vec{0}\}$ is a decreasing set and since both $S_{\textrm{top}}$ and $S_{\textrm{right}}$ are subsets of $\mathcal{Y}$, it follows that $Z \subseteq \mathcal{Y}$.

It is straightforward to computationally verify --- indeed, this was the reason for our choice of $p^*$ --- that $\nu^\mathcal{Y}(Z) = \mu^{\mathcal{X}_0}(\{\vec{0}\})$. Furthermore, using Theorem~\ref{regionthm} it is straightforward to verify that $\mu^{\mathcal{X}}|_{\mathcal{W}} \succeq \nu^{\mathcal{Y}}|_{\mathcal{W}}$.\footnote{The application of Theorem~\ref{regionthm} is made simpler by the fact that for $z_1 \in [.16016,.62307]$, the curve of $s$ lies below both $S_{\textrm{top}}$ and $S_{\textrm{right}}$.} Therefore, we can directly apply Theorem~\ref{generalthm} to this scenario to determine the optimal mechanism.

\begin{example}\label{betaexample}
The optimal mechanism for selling independent goods whose valuations are distributed according to $ f_1(z_1) = \frac{z_1^2(1-z_1)^2}{B(3,3)}$ and $f_2(z_2) = \frac{z_2^2(1-z_2)^3}{B(3,4)}$  has the following outcome for a bidder of type $(z_1,z_2)$ in terms of the function $s(\cdot)$ defined above:
\begin{itemize}
	\item If $(z_1,z_2) \in Z$, the bidder receives no goods and is charged 0.
	\item If $(z_1,z_2) \in \mathcal{A}$, the bidder receives item 1 with probability $-s'(z_1)$, item 2 with probability 1, and is charged $s(z_1) - z_1s'(z_1)$.
	\item If $(z_1, z_2) \in \mathcal{B}$, the bidder receives item 1 with probability 1, item 2 with probability $-1/s'(s^{-1}(z_2))$, and is charged $s^{-1}(z_2) - z_2/s'(s^{-1}(z_2))$.
	\item If $(z_1,z_2) \in \mathcal{W}$, the bidder receives both goods with probability 1 and is charged $.71307$.
\end{itemize}
\end{example}
Since $s(z_1)$ is not linear for $z_1 \in [0,.16016]$  and $z_1 \in [.62307,.63718]$, Example~\ref{betaexample}  shows that an optimal mechanism might offer a continuum of randomized outcomes.

%\begin{acks}
%The authors would like to thank Alessio Figalli for a helpful discussion about optimal transport theory.
%\end{acks}

%\bibliographystyle{chicago}
\bibliographystyle{abbrv}
\bibliography{costasbib}

\appendix

\section{Verifying Stochastic Dominance}

\subsection{An Equivalent Condition for Stochastic Dominance}

While Strassen's theorem is useful, it may be difficult to directly verify that a measure $\alpha$ stochastically dominates another measure $\beta$. Instead, we can check an equivalent condition, given by Lemma~\ref{finiteunions}. In preparation to state this lemma, we need a few claims and definitions.

%\begin{definition}
%A set $A \subseteq \mathbb{R}^n_{\geq 0}$ is increasing if $y \in A$ and $y \preceq x$ implies that $x \in A$.
%\end{definition}

\begin{claim}
{Let $\alpha, \beta$ be finite measures on $\mathbb{R}^n_{\ge 0}$.} A necessary and sufficient condition for $\beta \preceq \alpha$ is that for all increasing measurable sets $A$, $\alpha(A) \geq \beta(A)$.
\end{claim}

\begin{proof}
{Without loss of generality assume that $\alpha({\mathbb{R}^n_{\ge 0}}) =1$.} 

It is obvious that the condition is necessary by considering the indicator function of $A$. To prove sufficiency, suppose that the condition holds and that on the contrary, $\alpha$ does not stochastically dominate $\beta$. Then there exists an increasing, bounded, measurable function $f$ such that
$$\int f d \beta - \int f d \alpha > 2^{-k+1}$$
for some positive integer $k$. Without loss of generality, we may assume that $f$ is nonnegative, by adding the constant of $f(0)$ to all values. We now define the function $\tilde{f}$ by point-wise rounding $f$ upwards to the nearest multiple of $2^{-k}$. Clearly $\tilde{f}$ is increasing, measurable, and bounded. Furthermore, we have
$$\int \tilde{f} d\beta- \int \tilde{f}d\alpha \geq \int f d \beta- \int f d \alpha - 2^{-k} > 2^{-k+1} - 2^{-k} > 0.$$

We notice, however, that $\tilde{f}$ can be decomposed into the weighted sum of indicator functions of increasing sets. Indeed, let $\{r_1,\ldots, r_m\}$ be the set of all values taken by $\tilde{f}$, where $r_1 > r_2 > \cdots > r_m$. We notice that, for any $s\in \{1,\ldots,m\}$, the set $A_s = \{z : \tilde{f}(z) \geq r_s\}$ is increasing and measurable. Therefore, we may write
$$\tilde{f} = \sum_{s=1}^m (r_s-r_{s-1}) I_s$$
where $I_s$ is the indicator function for $A_s$  and where we set $r_0 = 0$. We now compute
$$\int \tilde{f} d\beta= \sum_{s=1}^m (r_s-r_{s-1})\beta(A_s) \leq \sum_{s=1}^m(r_s - r_{s-1}) \alpha(A_s) = \int \tilde{f} d\alpha,$$
contradicting the fact that $\int \tilde{f} d \beta > \int \tilde{f} d \alpha$.
\end{proof}

Given the above claim, to verify that  $\mu$ stochastically dominates  $\nu$, we must ensure that $\mu(A) \geq \nu(A)$ for all increasing measurable sets $A$. This verification might still be difficult, since $A$ has somewhat unconstrained structure. Our aim is to prove Lemma~\ref{finiteunions}, which will simplify this task further.

\begin{definition}
For any $z \in \mathbb{R}^n_{\geq 0}$, we define the \emph{base rooted at $z$} to be
$$B_z \triangleq \{ z': z \preceq z'\},$$
the minimal increasing set containing $z$. 
\end{definition}
We denote by $Q_k$ to be the set of points in $\mathbb{R}^n_{\geq 0}$ with all coordinates multiples of $2^{-k}$.

\begin{definition}
An increasing set $S$ is \emph{$k$-discretized} if $S = \bigcup_{z \in S \cap Q_k} B_z$. A \emph{corner} $c$ of a $k$-discretized set $S$ is a point $c \in S \cap Q_k$ such that there does not exist  $z \in S\setminus \{c\}$ with $z \preceq c$.
\end{definition}

\begin{lemma}
Every $k$-discretized set $S$ has only finitely many corners. Furthermore, $S = \cup_{c \in \mathcal{C}}B_c$, where $\mathcal{C}$ is the collection of corners of $S$.
\end{lemma}

\begin{proof}
We prove that there are finitely many corners by induction on the dimension, $n$. In the case $n=1$ the result is obvious, since if $S$ is nonempty it has exactly one corner. Now suppose $S$ has dimension $n$. Pick some corner $\hat{c} = (c_1, \ldots, c_n) \in S$. We know that any other corner must be strictly less than $\hat{c}$ in some coordinate. Therefore,
$$|\mathcal{C}| \leq {1 +} \sum_{i=1}^n \left| \left\{c \in \mathcal{C} \textrm{ s.t. } c_i < \hat{c}_i \right\}  \right| = {1 +} \sum_{i=1}^n \sum_{j = 1}^{2^k\hat{c}_i} \left| c \in \mathcal{C} \textrm{ s.t. } c_i = \hat{c}_i-2^{-k}j  \right|.$$
By the inductive hypothesis, we know that each set $\left\{ c \in \mathcal{C} \textrm{ s.t. } c_i = \hat{c}_i-2^{-k}j  \right\}$ is finite, since it is contained in the set of corners of the $(n-1)$-dimensional {subset of $S$ whose points have $i^{th}$ coordinate $\hat{c}_i - 2^{-k}j$.} Therefore, $|\mathcal{C}| $ is finite.

To show that $S =  \bigcup_{c \in \mathcal{C}}B_c$, pick any $z \in S$. Since $S$ is $k$-discretized, there exists a $b \in S \cap Q_k$ such that $z \in B_b$. If $b$ is a corner, then $z$ is clearly contained in $\bigcup_{c \in \mathcal{C}}B_c$. If $b$ is not a corner, then there is some other point $b' \in S \cap Q_k$ with $b' \preceq b$. If $b'$ is a corner, we're done. Otherwise, we repeat this process at most $2^k \sum_j b_j$ times, after which time we will have reached a corner $c$ of $S$. By construction, we have $z \in B_c$, as desired.\end{proof}

We now show that, to verify that one measure dominates another on all increasing sets, it suffices to verify that this holds for all sets that are the union of finitely many bases.

\begin{lemma}\label{finiteunions}
Let $g, h : \mathbb{R}^n_{\geq 0} \rightarrow \mathbb{R}_{\geq 0}$ be bounded density functions such that $\int_{\mathbb{R}^n_{\geq 0}} g(\vec{x}) d \vec{x}$ and $\int_{\mathbb{R}^n_{\geq 0}} h(\vec{x})d\vec{x}$ are finite. Suppose that, for all finite collections $Z$ of points in $\mathbb{R}^n_{\geq 0}$, we have
$$\int_{\bigcup_{z \in Z}B_z}g(\vec{x})d\vec{x} \geq \int_{\bigcup_{z \in Z}B_z}h(\vec{x})d\vec{x}.$$
Then, for all increasing sets $A$,
$$\int_A g(\vec{x})d\vec{x} \geq \int_A h(\vec{x})d\vec{x}.$$
\end{lemma}

\begin{proof}
Let $A$ be an increasing set. We clearly have $A = \bigcup_{z \in A} B_z$. For any point $z \in \mathbb{R}^n_{\geq 0}$, denote by $z^{n,k}$ the point in $\mathbb{R}^n_{\geq 0}$ whose $i^{th}$ component is the maximum of 0 and $z_i - 2^{-k}$ for each $i$.

We define
$$A_k^l \triangleq \bigcup_{z \in A \cap Q_k} B_z; \qquad A_k^u \triangleq \bigcup_{z \in A \cap Q_k} B_{z^{n,k}}.$$
It is clear that both $A_k^l$ and $A_k^u$ are $k$-discretized. Furthermore, for any $z \in A$ there exists a $z' \in A \cap Q_k$ such that each component of $z'$ is at most $2^{-k}$ more than the corresponding component of $z$. Therefore
$$A_k^l \subseteq A \subseteq A_k^u.$$

We now will bound
$$\int_{A_k^u} g d\vec{x} - \int_{A_k^l} g d\vec{x}.$$
Let
$$W_k = \left\{z : z_i > k \textrm{ for some } i   \right\}; \qquad W^c_k = \left\{z : z_i \leq k \textrm{ for all } i \right\}$$
We notice that
$$\int_{A_k^u \cap W_k} g d\vec{x} - \int_{A_k^l \cap W_k}g d\vec{x} \leq \int_{W_k} g d\vec{x}.$$
Furthermore, since $\lim_{k \rightarrow \infty}\int_{W^c_k}g d\vec{x} = \int_{\mathbb{R}^n_{\geq 0}}g d\vec{x}$,  we know that $\lim_{k \rightarrow \infty} \int_{W_k}g d\vec{x} = 0$. Therefore, 
$$\lim_{k \rightarrow \infty}\left(\int_{A_k^u \cap W_k} g d\vec{x} - \int_{A_k^l \cap W_k}g d\vec{x}\right) = 0.$$

Next, we bound
$$\int_{A_k^u \cap W^c_k} g d\vec{x} - \int_{A_k^l \cap W^c_k}g d\vec{x} \leq |g|_{\sup} \left(V(A_k^u \cap W^c_k) - V(A_k^l \cap W^c_k)  \right)$$
where $|g|_{\sup} < \infty$ is the supremum of $g$, and $V(\cdot)$ denotes the Lebesgue measure.

For each $m \in \{1,\ldots, n+1\}$ and  $z \in \mathbb{R}^n_{\geq 0}$, we define the point $z^{m,k}$ by:
$$z^{m,k}_i =
\begin{cases}
\max\{0, z_i - 2^{-k} \} & \mbox{ if } i < m\\
z_i & \mbox{ otherwise}
\end{cases}$$
and set
$$A_k^m \triangleq \bigcup_{z \in A \cap Q_k}B_{z^{m,k}}.$$
We have, by construction, $A_k^l = A_k^1$ and $A_k^u = A_k^{n+1}$. Therefore,
$$V(A_k^u \cap W^c_k) - V(A_k^l \cap W^c_k) = \sum_{m=1}^n \left(V(A_k^{m+1} \cap W_k^c)-V(A_k^{m} \cap W_k^c) \right).$$
We notice that, for any point $ (z_1,z_2, \ldots, z_{m-1},z_{m+1},\ldots, z_n) \in [0,k]^{n-1}$, there is an interval $I$ of length at most $2^{-k}$ such that the point $$(z_1,z_2,\ldots,z_{m-1},w,z_{m-2},\ldots,z_n) \in (A_k^{m+1} \setminus A_k^{m}) \cap W_k^c$$
if and only if $w \in I$. Therefore,
$$V(A_k^{m+1} \cap W_k^c)-V(A_k^{m} \cap W_k^c) \leq \int_0^k \cdots \int_0^k \int_0^k \cdots \int_0^k 2^{-k}dz_1\cdots dz_{m-1}dz_{m+1}\cdots dz_n = 2^{-k}k^{n-1}.$$

Therefore, we have the bound
$$
|g|_{\sup} \left(V(A_k^u \cap W^c_k) - V(A_k^l \cap W^c_k)  \right) \leq
|g|_{\sup}\sum_{m=1}^n 2^{-k}k^{n-1} = n|g|_{\sup}2^{-k}k^{n-1}
$$
and thus
\begin{align*}
\int_{A_k^u} g  d\vec{x} - \int_{A_k^l} gd\vec{x}  &= \int_{A_k^u \cap W_k}g d\vec{x} - \int_{A_k^l \cap W_k} gd\vec{x}  + \int_{A_k^u \cap W^c_k} g d\vec{x} - \int_{A_k^l \cap W^c_k} g d\vec{x} \\
&\leq  \left( \int_{A_k^u \cap W_k}g d\vec{x} - \int_{A_k^l \cap W_k} gd\vec{x}  \right) + n|g|_{\sup}2^{-k}k^{n-1}.
\end{align*}
In particular, we have
$$\lim_{k \rightarrow \infty} \left(\int_{A_k^u} g d\vec{x} - \int_{A_k^l} g d\vec{x}\right) = 0.$$
Since $\int_{A_k^u} g d\vec{x} \geq \int_A g d\vec{x} \geq  \int_{A_k^l}g d\vec{x}$, we have
$$\lim_{k \rightarrow \infty} \int_{A_k^u}g d\vec{x} = \int_A g d\vec{x} = \lim_{k \rightarrow \infty} \int_{A_k^l}g d\vec{x}.$$
Similarly, we have
$$\int_A h d\vec{x} = \lim_{k \rightarrow \infty}\int_{A_k^l} h d\vec{x}$$
and thus
$$\int_A (g-h) d\vec{x} = \lim_{k \rightarrow \infty}\left(\int_{A_k^l} g d\vec{x} - \int_{A_k^l} h d\vec{x} \right).$$
Since $A_k^l$ is $k$-discretized, it has finitely many corners. Letting $Z_k$ denote the corners of $A_k^l$, we have $A_k^l = \bigcup_{z \in Z_k}B_z$, and thus by our assumption $\int_{A_k^l}gd\vec{x} - \int_{A_k^l}hd\vec{x} \geq 0$
for all $k$. Therefore $\int_A(g-h)d\vec{x} \geq 0$, as desired.
\end{proof}
As an immediate corollary of Lemma~\ref{finiteunions}, we see that to verify $\nu \preceq \mu$ it suffices to check that $\nu(B) \leq \mu(B)$ for all sets $B$ which are unions of finitely many bases.

\subsection{Stochastic Dominance in Two Dimensions: Proof of Theorem~\ref{regionthm}}

In this section, we prove Theorem~\ref{regionthm}, which is a useful sufficient condition for stochastic dominance in two dimensions. To recap, we have $\mathcal{C} = [c_1, d_1^+) \times [c_2, d_2^+)$, $R$ a decreasing nonempty subset of $\mathcal{C}$, and $g, h : \mathcal{C} \rightarrow \mathbb{R}_{\geq 0}$ are bounded density functions which are $0$ on $R$, have finite total mass and satisfy
\begin{itemize}
	\item $\int_{\mathcal{C}} (g - h)dxdy \geq 0$.
	\item For any basis vector $e_i \in \{(0,1),(1,0) \}$ and any point $z^* \in R$:
	$$\int_0^{d_i^+ - z^*_i} g(z^* + \tau e_i) - h(z^* + \tau e_i)d\tau \leq 0.$$
	\item There exist non-negative functions $\alpha:  [c_1, d_1^+) \rightarrow \mathbb{R}_{\geq 0}$ and $\beta :  [c_2, d_2^+) \rightarrow \mathbb{R}_{\geq 0}$ and an increasing function $\eta : \mathcal{C} \rightarrow \mathbb{R}$ such that
	$$g(z_1,z_2) - h(z_1,z_2) = \alpha(z_1)\cdot \beta(z_2) \cdot \eta(z_1,z_2)$$
	for all $(z_1,z_2) \in \mathcal{C} \setminus R$.
\end{itemize}
We aim to prove that $g \succeq h$.

We begin by defining, for any  $c_1 \leq a \leq b \leq d_1^+$, the function $\zeta_a^b : [c_2,d_2^+) \rightarrow \mathbb{R}$ by
$$\zeta_a^b(w) \triangleq \int_a^b(g(z_1,w)-h(z_1,w))dz_1.$$
This function represents, for each $w$, the integral of $g - h$ along the line from $(a,w)$ to $(b,w)$.% along the line between $(c_1,a)$ and $(c_1, b)$.
\begin{claim}
If $(a,w) \in R$, then $\zeta_a^b(w) \leq 0$.
\end{claim}
\begin{proof}
The inequality trivially holds unless there exists a $z_1 \in [a, b]$ such that $g(z_1,w) > h(z_1,w)$. So suppose such a $z_1$ exists. It must be that $(z_1,w) \notin R$ as both $g$ and $h$ are $0$ in $R$. Indeed, because $R$ is a decreasing set it is also true that $(\tilde{z}_1,w) \notin R$ for all $\tilde{z}_1 \geq z_1$. This implies by assumption that
$$g(\tilde{z}_1,w) - h(\tilde{z}_1,w) = \alpha(\tilde{z}_1)\cdot \beta(w) \cdot \eta(\tilde{z}_1,w),$$
for all $\tilde{z}_1 \geq z_1$. Now given that $g(z_1,w) > h(z_1,w)$ and $\eta(\cdot,w)$ is an increasing function, we get that $g(\tilde{z}_1,w) \geq h(\tilde{z}_1,w)$ for all $\tilde{z}_1 \geq z_1$. Therefore, we have
$$\zeta_a^{z_1}(w)\leq \zeta_a^b(w) \leq \zeta_a^{d_1^+}(w).$$
We notice, however, that $\zeta_a^{d_1^+}(w) \leq 0$ by assumption, and thus the claim is proven.\end{proof}

%We know by assumption that
%$$\zeta_a^\infty(w) = \int_a^\infty (g(z_1,w) - h(z_1,w))dz_1 \leq 0.$$
We now claim the following:
\begin{claim}\label{signflip}
Suppose that $\zeta_a^b(w^*) > 0$ for some $w^* \in [c_2, d_2^+)$. Then $\zeta_a^b(w) \geq 0$ for all $w \in [w^*, d_2^+)$. 
\end{claim}
\begin{proof}
Given that $\zeta_a^b(w^*) > 0$, our previous claim implies that $(a,w^*)\not\in R$. Furthermore, since $R$ is a decreasing set and $w \geq w^*$, follows that $(a,w) \not\in R$, and furthermore that $(c,w) \not\in R$ for any $c \geq a$ in $[c_1, d_1^+)$. Therefore, we may write
%Since $\zeta_a^b(w^*) > 0$, there must exist a $z_1' \in [a,b]$ such that $g(z_1',w) > h(z_1',w)$. It follows by assumption that $g(\tilde{z}_1,w)\geq h(\tilde{z}_1,w)$ for all $\tilde{z}_1 \geq z_1'$, and thus
%$$\zeta_a^b(w) \geq \zeta_a^{z_1'}(w)$$
%
$$\zeta_a^b(w) = \int_a^b  (g(z_1,w) - h(z_1,w))dz_1 = \int_a^b (\alpha(z_1)\cdot \beta(w)\cdot \eta(z_1,w))dz_1.$$
Similarly, $(c,w^*) \not\in R$ for any $c \geq a$, so
$$\zeta_a^b(w^*) = \int_a^b (\alpha(z_1)\cdot \beta(w^*)\cdot \eta(z_1,w^*))dz_1.$$
Note that, since $\zeta_a^b(w^*) > 0$, we have $\beta(w^*) > 0$.
Thus, since $\eta$ is increasing,
\begin{align*}
\zeta_a^b(w) &\geq \int_a^b (\alpha(z_1)\cdot \beta(w) \cdot \eta(z_1, w^*))dz_1= \frac{\beta(w)}{\beta(w^*)} \zeta_a^b(w^*) \geq 0,
\end{align*}
as desired.\end{proof}

We extend $g$ and $h$ to all of $\mathbb{R}^2_{\geq 0}$ by setting them to be 0 outside of $\mathcal{C}$. By Claim~\ref{finiteunions}, to prove that $g \succeq h$ it suffices to prove that $\int_A g dxdy \geq \int_A h dxdy$ for all sets $A$ which are the union of finitely many bases. Since $g$ and $h$ are 0 outside of $\mathcal{C}$, it  suffices to consider only bases $B_{z'}$ where $z' \in \mathcal{C}$, since otherwise we can either remove the base (if it is disjoint from $\mathcal{C}$) or can increase the coordinates of $z'$ moving it to $\cal C$ without affecting the value of either integral.

We now prove Theorem~\ref{regionthm} by induction on the number of bases in the union.
\begin{itemize}
	\item \textbf{Base Case.}
	
	We aim to show $\int_{B_r} (g-h)dxdy \geq 0$ for any $r = (r_1,r_2) \in \mathcal{C}$.  We have
	\begin{align*}
	\int_{B_r} (g-h)dxdy = \int_{r_2}^{d_2^+} \int_{r_1}^{d_1^+} (g-h)dz_1dz_2 = \int_{r_2}^{d_2^+} \zeta_{r_1}^{d_1^+}(z_2)dz_2.
	\end{align*}
	By Claim~\ref{signflip}, we know that either $\zeta_{r_1}^{d_1^+}(z_2) \geq 0$ for all $z_2 \geq r_2$, or $\zeta_{r_1}^{d_1^+}(z_2) \leq 0$ for all $z_2$ between $c_2$ and $r_2$. In the first case, the integral is clearly nonnegative, so we may assume that we are in the second case. We then have
	\begin{align*}
	\int_{r_2}^{d_2^+} \zeta_{r_1}^{d_1^+}(z_2)dz_2 \geq \int_{c_2}^{d_2^+} \zeta_{r_1}^{d_1^+}(z_2)dz_2 =  \int_{c_2}^{d_2^+} \int_{r_1}^{d_1^+} (g-h)dz_1dz_2
	= \int_{r_1}^{d_1^+} \int_{c_2}^{d_2^+} (g-h)dz_2dz_1.
	\end{align*}
By an analogous argument to that above, we know that either $\int_{c_2}^{d_2^+} (g-h)(z_1,z_2)dz_2$ is nonnegative for all $z_1 \geq r_1$ (in which case the desired inequality holds trivially) or is nonpositive for all $z_1$ between $c_1$ and $r_1$. We assume therefore that we are in the second case, and thus
\begin{align*}
\int_{r_1}^{d_1^+} \int_{c_2}^{d_2^+} (g-h)dz_2dz_1 \geq \int_{c_1}^{d_1^+} \int_{c_2}^{d_2^+} (g-h)dz_2dz_1 = \int_{\mathcal{C}}(g-h)dxdy,
\end{align*}
which is nonnegative by assumption.	
	\item \textbf{Inductive Step.}
	Suppose that we have proven the result for all finite unions of at most $k$ bases. Consider now a set
$$A =  \bigcup_{i=1}^{k+1}  B_{z^{(i)}}.$$
We may assume that all $z^{(i)}$ are distinct, and that there do not exist distinct $z^{(i)}$, $z^{(j)}$ with $z^{(i)} \preceq z^{(j)}$, since otherwise we could remove one such $B_{z^{(i)}}$ from the union without affecting the set $A$, and the desired inequality would follow from the inductive hypothesis.

Therefore, we may order the $z^{(i)}$ such that 
$$c_1\leq z^{(k+1)}_1 < z^{(k)}_1< z^{(k-1)}_1 < \cdots < z^{(1)}_1$$ 
and
$$c_2 \leq z^{(1)}_2 < z^{(2)}_2 < z^{(3)}_2 < \cdots < z^{(k+1)}_2.$$ 

\begin{figure}[h!]
\begin{center}
\begin{tikzpicture}

\begin{axis}[ymin=1.1, ymax=3.3, xmin=0, xmax=1, xlabel=$z_1$, ylabel=$z_2$,  ytick pos=left, ytick={1.1}, yticklabels={$c_2$}, xtick={0},xticklabels={$c_1$}]
%xtick pos=left, xtick={.66666,1}, xticklabels={$\frac{2}{\lambda_1}$, $\frac{3}{\lambda_1}$},

%\pgfplotsset{ticks=none}
%\begin{axis}[ymin=0, ymax=4, xmin=0, xmax=4, xlabel=$x_1$, ylabel=$y1$, extra  x ticks={.66666,1}, extra x tick labels={$\frac{2}{\lambda_1}$, $\frac{3}{\lambda_1}$}, extra y ticks={2,3}, extra y tick labels={$\frac{2}{\lambda_2}$,$\frac{3}{\lambda_2}$}]
%\lambda_1=3, \lambda_2=1

%\addplot+[color=black,  fill=gray!20, mark=none, domain=0:1.5, samples=2]%area legend]
%{3 - 3*x}
%\closedcycle;
%\addlegendentry{$\mathcal{Y}$}

% \addplot+[color=gray, fill=gray, domain=.3:1.5, mark=none]
% {2-3*x}
% \closedcycle;
% \addlegendentry{$Z_{p^*}$}
 
  \addplot+[color=gray, fill=gray!50, domain=0:2,mark=none]
 %{1.9-1.5*x}
 {1.9-2*x*x}
 \closedcycle;

\addplot[color=black, mark=*] coordinates{
(.1,2.3)
(.1,7)
};

\addplot[color=black, mark=none] coordinates{
(.1,2.3)
(.3,2.3)
};

\addplot[color=black, mark=none] coordinates{
(.3,2.3)
(.3,1.6)
};

\addplot[color=black, mark=*] coordinates{
(.3,1.6)
};

\addplot[color=black, mark=none] coordinates{
(.3,1.6)
(.5,1.6)
};

\addplot[color=black, mark=none] coordinates{
(.5,1.6)
(.5,1.4)
};

\addplot[color=black, mark=*] coordinates{
(.5,1.4)};

\addplot[color=black, mark=none] coordinates{
(.5,1.4)
(7,1.4)
};

\node at (axis cs:.1,2.2){$z^{(k+1)}$};
\node at (axis cs:.3,1.5){$z^{(k)}$};
\node at (axis cs:.5,1.3){$z^{(k-1)}$};
\node at (axis cs:0.15,1.4){$R$};

\end{axis}
\end{tikzpicture}
\end{center}
\caption{We show that either decreasing $z^{(k+1)}_2$ to $z^{(k)}_2$ or removing $z^{(k+1)}$ entirely decreases the value of $\int_A (f-g)$. In either case, we can apply our inductive hypothesis.}
\end{figure}
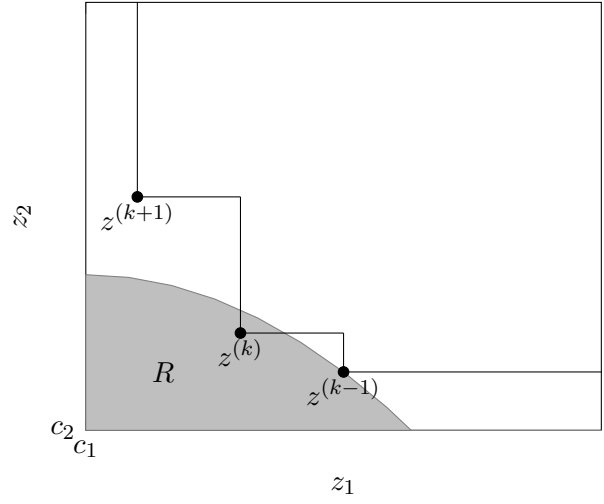

By Claim~\ref{signflip}, we know that one of the two following cases must hold:
\begin{itemize}
	\item \textbf{Case 1:} $\zeta_{z_1^{(k+1)}}^{z_1^{(k)}}(w) \leq 0$ for all $c_2 \leq w \leq z_2^{(k+1)}$.
	
	In this case, we see that 
	$$\int_{z_2^{(k)}}^{z_2^{(k+1)}} \int_{z_1^{(k+1)}}^{z_1^{(k)}} (f-g)dz_1dz_2 = \int_{z_2^{(k)}}^{z_2^{(k+1)}}  \zeta_{z_1^{(k+1)}}^{z_1^{(k)}}(w)dw \leq 0.$$
	For notational purposes, we will denote by $(f-g)(S)$ the integral $\int_S (f-g)dz_1dz_2$ for any set $S$. We now have
	\begin{align*}
	(f-g)(A) & \geq (f-g)(A)+ (f-g)\left(\left\{z: z_1^{(k+1)}\leq z_1 \leq z_1^{(k)} \textrm{ and } z_2^{(k)}\leq z_2 \leq z_2^{(k+1)} \right\} \right)\\
	&= (f-g)\left({ \bigcup_{i=1}^{k}  B_{z^{(i)}}  \cup B_{(z^{(k+1)}_1,z_2^{(k)})} }\right)\\
	&= (f-g)\left({ \bigcup_{i=1}^{k-1}  B_{z^{(i)}}  \cup B_{(z^{(k+1)}_1,z_2^{(k)})} }\right)
	\end{align*}
	where the last equality follows from $(z_1^{(k)},z_2^{(k)}) \succeq (z^{(k+1)}_1,z_2^{(k)})$. Now the inductive hypothesis implies that the quantity in the last line of the above derivation is $\ge 0$.

	\item \textbf{Case 2:} $\zeta_{z_1^{(k+1)}}^{z_1^{(k)}}(w) \geq 0$ for all $w \geq z_2^{(k+1)}$.
	
	In this case, we have
	$$\int_{z_2^{(k+1)}}^{d_2^+} \int_{z_1^{(k+1)}}^{z_1^{(k)}}(f-g)dz_1dz_2 = \int_{z_2^{(k+1)}}^{d_2^+}  \zeta_{z_1^{(k+1)}}^{z_1^{(k)}}(w)dw \geq 0.$$
	Therefore, it follows that
	\begin{align*}
	(f-g)(A) &= (f-g)\left(\bigcup_{i=1}^{k}  B_{z^{(i)}}  \right) + (f-g)\left( \left\{z: z_1^{(k+1)}\leq z_1 \leq z_1^{(k)} \textrm{ and }   z_2^{(k+1)} \leq z_2 \right\}  \right)\\
	&\geq (f-g)\left(\bigcup_{i=1}^{k}  B_{z^{(i)}}  \right) \geq 0,
	\end{align*}
where the final inequality follows from the inductive hypothesis.
\end{itemize}
\end{itemize}

\section{Bundling Two Power-Law Items}\label{powerappendix}

In this section, we complete the proof that a take-it-or-leave-it offer of the grand bundle is the optimal mechanism for selling two power-law items with $c_1 = 6$ and $c_2 = 7$. Based on our calculations in Section~\ref{powerlaw},  it remains to show that $\mu^{\mathcal{X}}|_{W_{p^*}} \succeq \nu^{\mathcal{Y}}|_{W_{p^*}}$, where $p^* = .35725$. We set $g = \mu_d^\mathcal{X}|_{W_{p^*}}$ and $h = \nu_d^\mathcal{Y}|_{W_{p^*}}$ and now must verify that the conditions of Theorem~\ref{regionthm} indeed apply.

The condition $\int_{\mathbb{R}^2_{\geq 0}} (g-h)dxdy = 0$ is satisfied by construction of $p^*$.

Since $\mu^\mathcal{X}$ and $\nu^{\mathcal{Y}}$ have disjoint supports we notice that for any $z \in W_{p^*}$ we have
$$\mu_d^\mathcal{X}(z) - \nu_d^{\mathcal{Y}}(z) = \left( \frac{c_1z_1}{1+z_1}+\frac{c_2z_2}{1+z_2} - 3 \right) f_1(z_1) \cdot f_2(z_2).$$
Thus, the third condition of Theorem~\ref{regionthm} is satisfied with $\alpha(z_1) = f_1(z_1)$, $\beta(z_2) = f_2(z_2)$, and $\eta(z_1,z_2) =  \frac{c_1z_1}{1+z_1}+\frac{c_2z_2}{1+z_2} - 3 $, noting that $\eta$ is indeed an increasing function.

All that remains is to verify the second condition of Theorem~\ref{regionthm}. We break this verification into two parts, depending on whether we are integrating with respect to $z_1$ or $z_2$.
\begin{itemize}
	\item We begin by considering integration with respect to $z_1$. That is, for any fixed $0 \leq z_2 \leq p^*$, we must prove that
	$$\int_{p^*-z_2}^\infty \left( \frac{c_1 z_1}{1+z_1} +\frac{c_2 z_2}{1+z_2} -3 \right)\frac{c_1 - 1}{(1+z_1)^{c_1}}\cdot \frac{c_2 - 1}{(1+z_2)^{c_2}} dz_1 \leq 0. $$
	Since $z_2$ is fixed, it clearly suffices to prove that
		$$\int_{p^*-z_2}^\infty \left(\frac{6 z_1}{1+z_1} + \frac{7 z_2}{1+z_2} -3 \right)\frac{1}{(1+z_1)^{6}}dz_1 \leq 0. $$
	This integral evaluates to
	$$\frac{-0.18565 + 1.1145z_2-2z_2^2}{(1.35725-z_2)^6(1+z_2)}.$$
	Since the denominator is always positive, it suffices to prove that the numerator is negative. Indeed, the numerator is maximized at $z_2 = .2786$, in which case the numerator evaluates to $-.0304$.
	\item We now consider integration with respect to $z_2$. Analogously to the computation above, for any fixed $0 \leq z_1 \leq p^*$ we must prove that 
		$$\int_{p^*-z_1}^\infty \left(\frac{6 z_1}{1+z_1} + \frac{7 z_2}{1+z_2} -3 \right)\frac{1}{(1+z_2)^{7}}dz_2 \leq 0. $$
		This integral evaluates to
		$$\frac{-0.0951667 + .595416z_1 - 1.66667z_1^2}{(1.35725 - z_1)^7(1+z_1)}.$$
		As before, it suffices to prove that the numerator is negative. We verify that, indeed, the numerator achieves its maximum at $z_1 = .178625$, in which case the numerator is $-.0419886$.
\end{itemize}
Therefore, we have proven, by Theorem~\ref{regionthm}, that  $\mu^{\mathcal{X}}|_{W_{p^*}} \succeq \nu^{\mathcal{Y}}|_{W_{p^*}}$, as desired.

\section{Optimal Mechanism Design in Two Dimensions: Proof of Theorem~\ref{generalthm}}\label{generalappendix}

Here we complete the proof of Theorem~\ref{generalthm}. We must show that a bidder of type $(z_1, z_2)$ never has incentive of falsely declaring a type $(z_1^*, z_2^*)$ in the following mechanism:
\begin{itemize}
	\item If $(z_1,z_2) \in Z$, the bidder receives no goods and is charged 0.
	\item If $(z_1,z_2) \in \mathcal{A}$, the bidder receives item 1 with probability $-s'(z_1)$, item 2 with probability 1, and is charged $s(z_1) - z_1s'(z_1)$.
	\item If $(z_1, z_2) \in \mathcal{B}$, the bidder receives item 1 with probability 1, item 2 with probability $-1/s'(s^{-1}(z_2))$, and is charged $s^{-1}(z_2) - z_2/s'(s^{-1}(z_2))$.
	\item If $(z_1,z_2) \in \mathcal{W}$, the bidder receives both goods with probability 1 and is charged $a + s(a)$.
\end{itemize}

The proof considers several cases.

\begin{itemize}
	\item \textbf{Case: } $(z_1^*,z_2^*) \in Z$.
	
	In this scenario, it will never be in the player's interest to change his bid to $(z_1^*, z_2^*)$. Indeed, he receives 0 utility by bidding in $Z$, while his utility of truthful bidding is never negative.
	
	\item \textbf{Case: } $(z_1, z_2) \in Z$.
	
	We first argue that the bidder has no incentive to deviate to a bid $(z_1^*,z_2^*) \in \mathcal{W}$. Indeed, since $(z_1, z_2) \in Z$, it follows that $z_1 + z_2 \leq a+s(a)$, and therefore by deviating to a bid in $W$ the player's new utility will not be positive.
	
	Similarly, we argue that the bidder has no incentive to bid $(z_1^*,z_2^*) \in \mathcal{A}$. By deviating, the bidder will receive utility
	\begin{align*}-z_1s'(z_1^*) + z_2 - s(z_1^*) + z_1^*s'(z_1^*) &= (z_1^*-z_1)s'(z_1^*) + (z_2 - s(z_1^*))\\
	&\leq (z_1^*-z_1)s'(z_1^*) + (s(z_1)-s(z_1^*)).
	\end{align*}
	We claim that this term utility is at most 0. Indeed, if $z_1 \leq z_1^*$,  then the first term is negative while the second term is positive, and the desired inequality follows by concavity of $s$. If $z_1^* \leq z_1$, then the first term is positive while the second is negative, and the result once again follows from concavity.
	
	The proof that the bidder has no incentive to bid $(z_1^*,z_2^*) \in \mathcal{B}$ is identical.
	
	\item \textbf{Case: } $(z_1,z_2)$ and $(z_1^*,z_2^*) \in \mathcal{W}$.
	
	By falsely declaring $(z_1^*,z_2^*)$, the bidder will still receive both goods for price of $a+s(a)$, and his utility will be unchanged.
	
	\item \textbf{Case: } $(z_1, z_2)$ and $(z_1^*, z_2^*) \in \mathcal{A}$.
	
	When truthful, the bidder receives utility $z_2 - s(z_1)$. By declaring $(z_1^*,z_2^*)$ instead, the bidder's utility instead is $z_2 - z_1s'(z_1^*) - s(z_1^*) + z_1^*s'(z_1^*)$. Thus, the difference between the truthful and non-truthful utilities is
	\begin{align*}
	z_2 - s(z_1) - \left( z_2 - z_1s'(z_1^*) - s(z_1^*) + z_1^*s'(z_1^*) \right) &= s(z_1^*)-s(z_1) +  s'(z_1^*)(z_1-z_1^*)\\
	&\geq 0
	\end{align*}
	where the final inequality follows from the identical argument to the case $(z_1,z_2) \in Z$.
	
	\item \textbf{Case: } $(z_1,z_2)$ and $(z_1^*,z_2^*) \in \mathcal{B}$.
	
	This case is identical to the case $(z_1,z_2)$ and $(z_1^*,z_2^*) \in \mathcal{A}$.
	
	\item \textbf{Case: } $(z_1, z_2) \in \mathcal{A}$ and $(z_1^*,z_2^*) \in \mathcal{W}$.
	
	This case is nearly identical to the case $(z_1, z_2)$ and $(z_1^*, z_2^*) \in \mathcal{A}$. Indeed, consider the function $s^*:\mathbb{R}_{\geq 0}\rightarrow \mathbb{R}$ which is identical to $s$ on $[0,b]$ but continues downwards at a $45^\circ$ angle on $[b,\infty]$.
	
	 We notice that, for $(z_1^*,z_2^*) \in \mathcal{W}$, we have $-{s^*}'(z_1^*) = 1$ and $$s^*(z_1^*)-z_1^*{s^*}'(z_1^*) = s^*(z_1) + z_1 = a+s(a).$$ Therefore, the exact same analysis holds as in the prior case, treating $(z_1^*,z_2^*)$ as being in $\mathcal{A}$ and where we replacing $s$ with $s^*$.
	 
	 \item \textbf{Case: } $(z_1,z_2) \in \mathcal{W}$ and $(z_1^*,z_2^*) \in \mathcal{A}$.
	 
	 We use an analogous argument to the previous case: after replacing $s^*$ for $s$, we can view $\mathcal{W}$ in the same manner as $\mathcal{A}$.
	 
	 \item \textbf{Case: } $(z_1, z_2) \in \mathcal{B}$ and $(z_1^*,z_2^*) \in \mathcal{W}$.
	 
	 This is analogous to the case $(z_1, z_2) \in \mathcal{A}$ and $(z_1^*,z_2^*) \in \mathcal{W}$.
	 
	 \item \textbf{Case: } $(z_1,z_2) \in \mathcal{W}$ and $(z_1^*,z_2^*) \in \mathcal{B}$.
	 
	 This is analogous to the case $(z_1,z_2) \in \mathcal{W}$ and $(z_1^*,z_2^*) \in \mathcal{A}$.
	 
	 \item \textbf{Case:} $(z_1,z_2) \in \mathcal{B}$ and $(z_1^*,z_2^*) \in \mathcal{A}$.
	 
	 The difference between truthfully bidding $(z_1,z_2)$ and falsely declaring $(z_1^*,z_2^*)$ is
	 \begin{align*}
	 \left( z_1 - s^{-1}(z_2)  \right) - \left(  -z_1s'(z_1^*) + z_2 - s(z_1^*) + z_1^*s'(z_1^*) \right)\\
	 = (s(z_1^*)-z_2) + (z_1 - s^{-1}(z_2)) - ( -s'(z_1^*)(z_1 - z_1^*)).
	 \end{align*}
	We aim to prove that the above quantity is nonnegative. Since $(z_1,z_2) \in \mathcal{B}$ and  $(z_1^*,z_2^*) \in \mathcal{A}$, we have the following inequalities:
	 \begin{align*}
	 z_1^* &< s^{-1}(z_2) \leq z_1\\
	 z_2 &< s(z_1^*) \leq z_2^*
	 \end{align*}
	 and therefore $ (s(z_1^*)-z_2) $, $(z_1 - s^{-1}(z_2))$, and $ ( -s'(z_1^*)(z_1 - z_1^*))$ are all positive. In particular,  the desired inequality is immediate if $s'(z_1^*)=0$, and we may therefore assume that $s'(z_1^*) < 0$.
	 
	 We may also assume that $( -s'(z_1^*)(z_1 - z_1^*)) > s(z_1^*) - z_2$, since otherwise the inequality is immediate. Therefore, there exists an $x$ in between $z_1^*$ and $z_1$ such that
	$$( -s'(z_1^*)(x - z_1^*)) = s(z_1^*) - z_2.$$
	We claim now that $x \geq s^{-1}(z_2)$. Indeed, since $s$ is convex and decreasing, we compute
	\begin{align*}
	-s'(z_1^*)(s^{-1}(z_2) - z_1^*) \leq \frac{s(z_1^*)-z_2}{s^{-1}(z_2)-z_1^*}\cdot (s^{-1}(z_2) - z_1^*)  = s(z_1^*)-z_2.
	\end{align*}
	 Since $ -s'(z_1^*)(x - z_1^*)$ is an increasing function of $x$, it follows that $ s^{-1}(z_2)\leq x \leq z_1$.
	 We now write the difference between bidding $(z_1,z_2)$ and falsely declaring $(z_1^*,z_2^*)$ as
	 \begin{align*}
	& (s(z_1^*)-z_2) + (z_1 - s^{-1}(z_2)) - ( -s'(z_1^*)(z_1 - z_1^*)) = \\
	 &(s(z_1^*)-z_2) + (z_1 - s^{-1}(z_2)) - ( -s'(z_1^*)(x - z_1^*)) - (-s'(z_1^*)(z_1 - x))=\\
	 & (z_1 - s^{-1}(z_2)) - (-s'(z_1^*)(z_1 - x)) \geq \\
	 & (z_1 - x) - (-s'(z_1^*)(z_1 - x)).
	 \end{align*}
	Since $(z_1^*, z_2^*) \in \mathcal{A}$, it follows that $-1 \leq s'(z_1^*) \leq 0$, and the above expression is therefore nonnegative, as desired.

\item \textbf{Case:} 	$(z_1,z_2) \in \mathcal{A}$ and $(z_1^*,z_2^*) \in \mathcal{B}$.

This is identical to the case $(z_1,z_2) \in \mathcal{B}$ and $(z_1^*,z_2^*) \in \mathcal{A}$.	
\end{itemize}
We have thus shown that the bidder never has incentive to deviate from his truthful strategy, and therefore the mechanism is truthful.

\end{document}